\let\accentvec\vec
\documentclass[runningheads, envcountsame, a4paper]{llncs}

\let\vec\accentvec
\usepackage[T1]{fontenc}
\usepackage{graphicx}
\usepackage{comment}

\usepackage{amsmath, amssymb}
\usepackage[utf8]{inputenc}

\usepackage{thmtools, thm-restate}

\usepackage[pagewise]{lineno}

\usepackage{hyperref}
\hypersetup{
    colorlinks,
    allcolors=blue,
    unicode,
    bookmarksopen,
    bookmarksdepth=2,
    pdftitle={Distance Recoloring},
    pdfauthor={Niranka Banerjee, Christian Engels, Duc A. Hoang},
    pdfkeywords={Reconfiguration problem, d-Distance k-Coloring, Computational complexity, PSPACE-completeness, Polynomial-time algorithm}
}

\usepackage[capitalise,nameinlink]{cleveref}
\crefname{claim}{Claim}{Claims}
\crefname{figure}{Figure}{Figures}
\usepackage{complexity}
\usepackage[dvipsnames]{xcolor}
\usepackage{tikz}
\usetikzlibrary{automata,positioning,arrows,shapes,calc,matrix,math,backgrounds,quotes,graphs,external}
\usepackage{algorithm}
\usepackage[indLines=false]{algpseudocodex}
\usepackage{xspace}
\usepackage[inline]{enumitem}

\newcommand{\CR}{\textsc{Coloring Reconfiguration}\xspace}
\newcommand{\ST}{\textsc{Sliding Tokens}\xspace}
\newcommand{\dist}{\mathsf{dist}}
\newcommand{\cl}{\mathsf{cl}}
\newcommand{\far}{\mathsf{far}}
\renewcommand{\part}{\mathsf{part}}
\newcommand{\diam}{\mathsf{diam}}

\usepackage[appendix=inline, bibliography=separate]{apxproof}\makeatletter
\g@addto@macro\thmt@newtheorem@postdefinition{    \let\thmt@theoremdefiner\thmt@original@spnewtheorem
}
\makeatother

\newtheoremrep{theorem}{Theorem}
\newtheoremrep{lemma}{Lemma}
\newtheoremrep{proposition}{Proposition}
\newtheoremrep{claim}{Claim}
\newtheoremrep{corollary}{Corollary}

\let\doendproof\endproof
\renewcommand\endproof{~\hfill$\qed$\doendproof}

\makeatletter
\renewcommand\paragraph{\@startsection{paragraph}{4}{\z@}                     {-12\p@ \@plus -4\p@ \@minus -4\p@}                     {-0.5em \@plus -0.22em \@minus -0.1em}                     {\normalfont\normalsize\bfseries}}
\makeatother

\usepackage[margin=2cm]{geometry}

\begin{document}
\title{Distance Recoloring}
\author{Niranka~Banerjee\inst{1} \and
Christian~Engels\inst{2}\and
Duc A.~Hoang\inst{3}}
\institute{Mie University, Tsu, Japan
\\
\email{banerjee@eng.mie-u.ac.jp}
\and
National Institute of Informatics, Tokyo, Japan \\
\email{christian.engels@gmail.com}\and
VNU University of Science, Vietnam National University, Hanoi, Vietnam\\
\email{hoanganhduc@hus.edu.vn}
}

\maketitle\begin{abstract}
Reconfiguration problems ask whether one feasible solution can be transformed into another by a sequence of local moves while maintaining feasibility throughout. For integers $d \geq 1$ and $k \geq d+1$, the \textsc{Distance Coloring} problem asks if a given graph $G$ has a $(d, k)$-coloring, i.e., a coloring of the vertices of $G$ by $k$ colors such that any two vertices within distance $d$ from each other have different colors.

For ordinary proper colorings ($d=1$), the $k$-Coloring Reconfiguration problem is polynomial-time solvable for $k\le 3$ [Cereceda, van den Heuvel, and Johnson, \emph{J. Graph Theory} 67(1):69--82, 2011] but is $\mathsf{PSPACE}$-complete for every fixed $k\ge 4$, even on bipartite graphs [Bonsma and Cereceda, \emph{Theor. Comput. Sci.} 410(50):5215--5226, 2009].
In this work, we initiate a study of the distance-$d$ analogue, for $d \geq 2$. We show that even for planar, bipartite, and $2$-degenerate graphs, \textsc{$(d, k)$-Coloring Reconfiguration} remains $\mathsf{PSPACE}$-complete for every $d \geq 3$ via a reduction from the well-known \textsc{Sliding Tokens} problem. Our construction uses $k = k_0 + 2 + n(\lceil d/2\rceil-1)$ colors on instances of size $n$, where $k_0\in\{3d+3,3d+6\}$ (depending on the parity of $d$). For $d = 2$, the same reduction scheme can be adapted to show that the problem is $\mathsf{PSPACE}$-complete on planar and $2$-degenerate graphs with same values of $k$.
Additionally, on split graphs, there is an interesting dichotomy: the problem is $\mathsf{PSPACE}$-complete when $d = 2$ and $k$ is large but can be solved efficiently when $d \geq 3$ and $k \geq d+1$. For chordal graphs, we show that the problem is $\mathsf{PSPACE}$-complete for even values of $d \geq 2$. Finally, we design a quadratic-time algorithm to solve the problem on paths for any $d \geq 2$ and $k \geq d+1$.

\end{abstract}
\section{Introduction}\label{sec:intro}
For the last few decades, \textit{reconfiguration problems} have emerged in various areas of computer science, including computational geometry, recreational mathematics, and constraint satisfaction~\cite{Heuvel13,Nishimura18,MynhardtN19}.
Given a \textit{source problem} $\mathcal{P}$ (e.g., \textsc{Satisfiability}, \textsc{Coloring}, \textsc{Independent Set}),
one can define its \textit{reconfiguration variants}.
In such a variant, two \textit{feasible solutions} (e.g., satisfying truth assignments, proper vertex-colorings, independent sets) $S$ and $T$ of $\mathcal{P}$ are given along with a prescribed \textit{reconfiguration rule} that usually describes a ``small'' change in a solution.
The question is to decide if there is a sequence of feasible solutions that transforms $S$ into $T$, where each intermediate member is obtained from its predecessor by applying the reconfiguration rule exactly once.
Such a sequence, if it exists, is called a \textit{reconfiguration sequence}.

\paragraph{Distance Coloring}\label{sec:intro-dc}

The concept of \textit{$(d, k)$-coloring} (or \textit{$d$-distance $k$-coloring}) was introduced in 1969 by Kramer and Kramer~\cite{KramerK1969,KramerK1969-2}.
For a graph $G = (V, E)$ and integers $d \geq 1$ and $k \geq d+1$, a \textit{$(d, k)$-coloring} of $G$ is an assignment of $k$ colors to the vertices of $G$ such that no two vertices within distance $d$ share the same color.
In particular, the classic \textit{proper $k$-coloring} is the case when $d = 1$.
The \textsc{$(d, k)$-Coloring} problem, which asks if a given graph $G$ has a $(d, k)$-coloring, is known to be \NP-complete for any fixed $d \geq 2$ and large $k$~\cite{Mccormick1983,LinS1995}.
In 2007, Sharp~\cite{Sharp07} proved the following complexity dichotomy: \textsc{$(d, k)$-Coloring} can be solved in polynomial time for $k \leq \lfloor 3d/2 \rfloor$ but is \NP-hard for $k > \lfloor 3d/2 \rfloor$.
We refer readers to the survey~\cite{KramerK2008} for more details on related \textsc{$(d, k)$-Coloring} problems.

\paragraph{{\bf Coloring Reconfiguration}}\label{intro:cr}
$k$-\CR has been extensively studied in the literature~\cite{Nishimura18,MynhardtN19,Mahmoud2024}.
In $k$-\CR, we are given two proper $k$-colorings $\alpha$ and $\beta$ of a graph $G$ and want to decide if
there exists a way to recolor vertices one by one, starting from $\alpha$ and ending at $\beta$,
such that every intermediate coloring is still a proper $k$-coloring.
It is well-known that $k$-\CR is \PSPACE-complete for any fixed $k \geq 4$ on bipartite graphs, for any fixed $4 \leq k \leq 6$ on planar graphs, 
and for $k = 4$ on bipartite planar graphs (and thus $3$-degenerate graphs)~\cite{BonsmaC09}. 
A further note from Bonsma and Paulusma~\cite{BonsmaP19} shows that $k$-\CR is \PSPACE-complete even on $(k-2)$-connected bipartite graphs for $k \geq 4$.
Indeed, the problem remains \PSPACE-complete even on planar graphs of bounded bandwidth and low maximum degree~\cite{vanderZanden15}.
On the other hand, for $1 \leq k \leq 3$, $k$-\CR can be solved in linear
time~\cite{CerecedaHJ11,JohnsonKKPP16}.
Additionally, $k$-\CR is solvable in polynomial time on planar graphs for $k \geq 7$ and on bipartite planar graphs for
$k \geq 5$~\cite{BonsmaC09,Heuvel13}.
With respect to graph classes, $k$-\CR is solvable in polynomial time on $2$-degenerate graphs (which contains graphs of treewidth at most
two such as trees, cacti,
outerplanar graphs, and series-parallel graphs) and several subclasses of chordal graphs~\cite{HatanakaIZ19,BonsmaP19}. 

\paragraph{{\bf List Coloring Reconfiguration}}\label{sec:intro-lcr}
A generalized variant of $k$-\CR, the \textsc{List $k$-\CR} problem, has also been
well-studied. Here, like in $k$-\CR, given a graph $G$ and two proper $k$-colorings $\alpha, \beta$, we want to transform
$\alpha$ into $\beta$ and vice versa.
However, we also require that each vertex has a list of at most $k$ colors from $\{1, \dots, k\}$ attached,
which are the only colors each vertex is allowed to have.
In particular, $k$-\CR is nothing but \textsc{List $k$-\CR} when every color list is $\{1, \dots, k\}$.
Indeed, along the way of proving the \PSPACE-completeness of $k$-\CR, Bonsma and Cereceda~\cite{BonsmaC09} showed that \textsc{List $k$-\CR} is \PSPACE-complete for any fixed $k \geq 4$.
Cereceda, Van den Heuvel, and Johnson~\cite{CerecedaHJ11} showed that $k$-\CR is solvable in polynomial time for $1 \leq k \leq 3$ and their algorithms can
be extended for \textsc{List $k$-\CR}.\footnote{Van den Heuvel~\cite{Heuvel13} stated that \textsc{$k$-List-Color-Path}    is \PSPACE-complete for any $k \geq 3$,
    which appears to be different from
    what we mentioned for $k = 3$.
    However, note that, the two problems are different.
    In his definition, each list has size at most $k$, but indeed one may use more than $k$ colors in total.
    On the other hand, in our definition, one \textit{cannot} use more than $k$ colors in total.}
Hatanaka, Ito, and Zhou~\cite{HatanakaIZ15} initiated a systematic study of \textsc{List $k$-\CR} and showed the following complexity dichotomy: The problem is
\PSPACE-complete on graphs of
pathwidth two but polynomial-time solvable on graphs of pathwidth one (whose components are caterpillars---the trees obtained by attaching leaves to a central path).
They also noted that their hardness result can be extended for threshold graphs.
Wrochna~\cite{Wrochna18} showed that \textsc{List $k$-\CR} is \PSPACE-complete on bounded bandwidth graphs and the constructed graph in his
reduction also has pathwidth two,
which independently confirmed the result of Hatanaka, Ito, and Zhou~\cite{HatanakaIZ15}.

\subsection{Our Problem and Results}
In this paper, for $d \geq 2$ and $k \geq d+1$, we study the $(d,k)$-\CR problem, a generalized variant of $k$-\CR.
In $(d, k)$-\CR, we are given two $(d, k)$-colorings $\alpha$ and $\beta$ of a graph $G$ and want to decide if
there exists a way to recolor vertices one by one, starting from $\alpha$ and ending at $\beta$,
such that every intermediate coloring is still a $(d, k)$-coloring.

Recall that the \textit{$d$-th power} of a graph $G$, denoted by $G^d$, is the graph with $V(G^d) = V(G)$ and
$E(G^d) = \{uv \mid u, v \in V(G) \text{ and } \dist_G(u, v) \leq d\}$, where $\dist_G(u, v)$ denotes the \textit{distance} (i.e., the length of a shortest path)
between $u$ and $v$ in $G$.
It is well known that $\alpha$ is a $(d, k)$-coloring of a graph $G$ if and only if $\alpha$ is a $(1, k)$-coloring of $G^d$.
At first glance, this relationship suggests a straightforward proof of \PSPACE-completeness for $(d,k)$-\CR: since the problem is already known to
be \PSPACE-complete for $d = 1$ and $k \geq 4$~\cite{BonsmaC09}, one might attempt to reduce $(1,k)$-\CR to $(d,k)$-\CR.
However, we emphasize that \textit{such a reduction fails}.

For a valid polynomial-time reduction, given any instance $(G, \alpha, \beta)$ of $(1, k)$-\CR, one would need to
construct a corresponding instance $(H, \alpha, \beta)$ of $(d, k)$-\CR where $H$ is a \textit{$d$-th root} of $G$ (i.e., $H^d \simeq G$).
The problem lies in finding such a root: deciding whether a graph $G$ has a $d$-th root is $\NP$-complete for all fixed $d \geq 2$ in chordal graphs and remains $\NP$-complete even on bipartite graphs for all fixed $d \geq 3$~\cite{LeN2010}.
Therefore, unless $\P = \NP$, this reduction cannot be computed in polynomial time, invalidating this approach for proving \PSPACE-completeness.

Our major contribution, shown in \cref{sec:completeness}, is that for every $d \geq 3$ and $k \geq d+1$, $(d,k)$-\CR is \PSPACE-complete even on graphs that are bipartite, planar, and $2$-degenerate. 
Additionally, when $d = 2$, the same reduction scheme can be adapted to show that the problem is \PSPACE-complete on planar and $2$-degenerate graphs.
In this context, recall that $k$-\CR was \PSPACE-complete on bipartite graphs for $k \geq 4$ and planar graphs for $4 \leq k \leq 6$ but was polynomial time solvable on planar graphs for $k \geq 7$ and $2$-degenerate graphs.

\begin{restatable}{theorem}{thmmain}\label{thm:main}
   
    For every integer $d \geq 2$, given two $(d,k)$-colorings $\alpha,\beta$ of a graph $G$ (where $k$ is part of the input), it is \PSPACE-complete to decide if there is a reconfiguration sequence that transforms $\alpha$ into $\beta$ and vice versa, even if either 
    \begin{itemize}
        \item[(i)] $d \geq 3$ and $G$ is bipartite, planar, and $2$-degenerate, or
        \item[(ii)] $d = 2$ and $G$ is planar and $2$-degenerate.
    \end{itemize}
    Moreover, in case~(i) our reduction produces instances with $k = k_0+2+n(\lceil d/2\rceil-1)=\Theta(nd + k_0)$ colors, where $k_0\in\{3d+3,3d+6\}$.\end{restatable}

In \cref{sec:split}, we investigate the $(d,k)$-\CR problem on split graphs and chordal graphs.
First, we revisit the \NP-completeness proof by Bodlaender et al.~\cite{BodlaenderKTL04} for \textsc{$(2, k)$-Coloring} on split graphs. 
Based on their reduction, we then establish that $(2,k)$-\CR is \PSPACE-complete on split graphs for sufficiently large values of $k$.

\begin{restatable}{theorem}{thmsplit}\label{thm:2kcolreconf-split}
    $(2, k)$-\CR on split graphs is \PSPACE-complete.
\end{restatable}

Additionally, we further extend the reduction on split graphs and show the following result for chordal graphs.

\begin{restatable}{theorem}{thmchordal}\label{thm:chordal}
    $(d, k)$-\CR is \PSPACE-complete on chordal graphs for even values of $d \geq 2$.
\end{restatable}

On the algorithmic side (\cref{sec:algorithms}), we show simple polynomial-time algorithms for graphs of diameter at most $d$ (\cref{sec:ddiameter}) and paths (\cref{sec:paths}).

\begin{restatable}{theorem}{thmdiamd}\label{thm:diam-d}
    Let $G$ be any $(d, k)$-colorable graph on $n$ vertices whose diameter is at most $d$.
    Then, $(d, k)$-\CR is solvable in $O(\log n + \log k)$ time.
    Moreover, given a yes-instance $(G, \alpha, \beta)$, one can construct in $O(n)$ time a reconfiguration sequence between $\alpha$ and
    $\beta$.
\end{restatable}

\begin{restatable}{theorem}{thmpaths}\label{thm:paths}
    $(d, k)$-\CR on $n$-vertex paths can be solved in $O(\log k + \log d)$ time.
    Moreover, in a yes-instance, one can construct a corresponding reconfiguration sequence in $O(n^2)$ time.
\end{restatable}

In particular, \cref{thm:diam-d} implies that $(d, k)$-\CR is in \P{} on split graphs (whose components have a diameter of at most $3$) for any $d \geq 3$.

\subsection{On the Technical Contributions of Our Main Result}
\label{sec:tech-contrib}

In this section, we outline the main technical obstacles that arise in proving our main theorem (\cref{thm:main}),
and we summarize the ideas used to overcome them.

\paragraph{Scope of the Main Hardness Statement}
Our primary contribution, presented in \cref{sec:completeness}, shows that for every fixed $d \ge 3$,
the $(d,k)$-\CR problem is \PSPACE-complete even when restricted to graphs that are bipartite, planar, and
$2$-degenerate, with $k$ given as part of the input (in particular, the number of colors produced by our reduction
is $k=\Theta(nd)$). For $d=2$, the same reduction scheme yields \PSPACE-completeness on planar $2$-degenerate graphs.
The case $d=1$ was settled by Bonsma and Cereceda~\cite{BonsmaC09}; our contribution is the extension to distance-$d$
constraints under strong structural restrictions.

\paragraph{Two-Step Strategy}
Our approach follows the high-level template of~\cite{BonsmaC09} but requires additional ideas to handle the fact that
constraints are imposed at distance $d\ge 2$ rather than only on adjacent vertices.
We proceed via two reductions.

\begin{enumerate}
  \item \textbf{List-to-nonlist reduction (\cref{sec:dk-col-reconf}).}
  We give a polynomial-time reduction from \textsc{List $(d,k)$-\CR} to $(d,k')$-\CR with $k' = k+2+n(\lceil d/2\rceil-1) \;=\; \Theta(nd+k)$,
  where $n$ is the number of vertices in the input list-instance.

  \item \textbf{Hardness of the list version (\cref{sec:list-dk-col-reconf}).}
  We show that \textsc{List $(d,k)$-\CR} is \PSPACE-complete on bipartite, planar, and $2$-degenerate graphs
  for every fixed $d\ge 2$ and sufficiently large fixed $k$
  (namely $k \ge 3d+3$ if $d$ is odd and $k \ge 3d+6$ if $d$ is even).
  This is obtained by first introducing, in \cref{sec:sliding-tokens}, a variant of the \ST problem that remains
  \PSPACE-complete even on planar $2$-degenerate graphs, and then giving a polynomial-time reduction from this
  variant to \textsc{List $(d,k)$-\CR}.

  Finally, combining these two steps yields our main \PSPACE-hardness statement for $(d,k)$-\CR by transforming
  the resulting list-instance into a non-list instance via the list-to-nonlist reduction of \cref{sec:dk-col-reconf}.
  In particular, this final transformation uses $k'=\Theta(nd+k)$ colors, and \cref{lem:dkcol-color-lb} shows that
  the dependence on $n$ is inherent for the present scheme.
\end{enumerate}

\paragraph{Why Distance-$d$ Creates New Difficulties}
To highlight the technical differences with~\cite{BonsmaC09}, we describe the key locality properties exploited
when $d=1$ and explain why they fail for $d\ge 2$.
For $j\in\{1,2\}$, let $G_j$ denote the graph of the source instance of Reduction~$(j)$ and let $G'_j$ denote the
graph of the produced target instance.

\paragraph{Reduction (1): Simulating Lists via Frozen Graphs}
In Reduction~(1), each vertex $v$ of $G_1$ is assigned a list of admissible colors $L(v)\subseteq \{1,\dots,k\}$,
and all vertices of $G_1$ remain present in $G'_1$.
Bonsma and Cereceda~\cite{BonsmaC09} enforce the restriction ``$c\notin L(v)$'' (for $d=1$) as follows:
for each pair $(v,c)$ with $c\notin L(v)$ they attach to $v$ a pre-colored frozen graph $F_{v,c}$ containing all $k$
colors, and they join $v$ to a vertex of $F_{v,c}$ colored $c$.
Since the coloring of $F_{v,c}$ is frozen, $v$ can never take color $c$, thereby simulating the forbidden color.
The crucial locality property is that, when $d=1$, the attachment vertex $v$ acts as a separator:
vertices of $G_1$ other than $v$ are not constrained by vertices inside $F_{v,c}$.

For $d\ge 2$, this locality is no longer automatic.
Indeed, if a vertex $w\in V(G_1)$ lies within distance $d$ of some vertex $x\in V(F_{v,c})$ colored $c_x$,
then the distance-$d$ rule forbids $w$ from receiving color $c_x$ in $G'_1$.
This may inadvertently remove colors that are allowed at $w$ in the original list instance (i.e., it may happen that
$c_x\in L(w)$), so the naive frozen-gadget attachment fails.
Consequently, to correctly simulate lists at distance $d$, we must (i) introduce additional auxiliary colors and
(ii) design the frozen gadgets so that they forbid exactly the intended colors and do not ``leak'' constraints to
unrelated vertices of $G_1$.

A second obstacle is global: once auxiliary colors are introduced, every frozen graph must use \emph{all} colors in the
palette (both original and auxiliary) in its frozen coloring.
Since our construction introduces auxiliary colors that depend on the vertex set of $G_1$, each frozen graph attached
to a vertex must be able to realize all these colors while remaining frozen, planar, bipartite (for $d\ge 3$),
and sparse, and while still ensuring that auxiliary colors cannot be used on the original vertices of $G_1$.
In \cref{sec:dk-col-reconf}, we present a dedicated family of frozen graphs that overcomes these issues.
In particular, unlike the frozen graphs in~\cite{BonsmaC09} (which are either bipartite or planar), our gadgets are
bipartite (for $d \geq 3$), planar, and even $2$-degenerate, while still supporting the required frozen
colorings under the distance-$d$ constraint.

\paragraph{Reduction (2): Enforcing Token Constraints via Forbidding Paths}
In Reduction~(2), Bonsma and Cereceda~\cite{BonsmaC09} reduce from \ST by encoding token placements as list colorings.
Their \ST instances consist of token triangles and token edges. Each token triangle $i$ (with vertices
$t_{i1},t_{i2},t_{i3}$) is contracted to a vertex $t_i$ in $G'_2$ with list $\{1,2,3\}$, and each token edge $j$
(with vertices $e_{j1},e_{j2}$) is contracted to a vertex $e_j$ with list $\{1,2\}$.
A token move within a gadget (e.g., from $t_{i1}$ to $t_{i2}$) corresponds to recoloring the corresponding vertex
(e.g., recoloring $t_i$ from $1$ to $2$).

Adjacency constraints between tokens are enforced using \emph{$(a,b)$-forbidding paths}:
if the vertex $t_{ia}$ of triangle $i$ is adjacent (in the \ST instance) to the vertex $t_{\ell b}$ of triangle $\ell$,
then one must ensure that $t_i$ and $t_\ell$ are never simultaneously colored $a$ and $b$.
This is achieved by placing an $(a,b)$-forbidding path between $t_i$ and $t_\ell$ whose internal vertices use auxiliary
colors outside $\{1,2,3\}$.
Again, for $d=1$ a key locality property holds: recoloring an internal vertex of a forbidding path does not affect
vertices outside the path, because the endpoints separate the internal auxiliary recolorings from the rest of the graph.

For $d\ge 2$, however, internal vertices of distinct forbidding paths can lie within distance $d$ in $G'_2$,
causing unintended interactions between gadgets.
This destroys the intended ``local simulation'' of token moves, unless one carefully arranges both graph distances and
auxiliary palettes.
In \cref{sec:list-dk-col-reconf}, we extend the forbidding-path framework by introducing buffer structure and a controlled
assignment of auxiliary colors so that recolorings inside one forbidding path remain local with respect to the
distance-$d$ constraint, even in the presence of many nearby gadgets.
This refinement is essential for the correctness of the reduction at distance $d$.

\paragraph{A Degeneracy Dichotomy}
Our constructions also imply \PSPACE-completeness on $2$-degenerate graphs for every $d\ge 2$.
Together with the known fact that the problem is in $\P$ for $d=1$ on $2$-degenerate graphs~\cite{HatanakaIZ19},
this yields an appealing complexity dichotomy on this sparse graph class.

\section{Preliminaries}\label{sec:prelims}

We refer readers to~\cite{Diestel2017} for the concepts and notations not defined here.
Unless otherwise mentioned, we always consider simple, undirected, connected graphs.
For two vertices $u, v$ of a graph $G$, we denote by $\dist_G(u, v)$ the \textit{distance}
(i.e., the length
of the shortest path) between $u$
and $v$ in $G$.
The \textit{diameter} of a graph $G$, denoted by $\diam(G)$, is the maximum distance between any pair of vertices in $G$.
We denote by $\Delta(G)$ and $\delta(G)$ the maximum and minimum degrees of a vertex of $G$, respectively.
We define $N_d(v)$ for a given graph $G$ to be the set of all vertices of distance at most $d$, i.e.,
$N_d(v) = \{ u\in V \mid \dist_G(u,v)\leq d\}$.
An \textit{$s$-degenerate} graph is an undirected graph in which every induced subgraph has a vertex of degree at most $s$.

\paragraph{$(d, k)$-Coloring}
For two positive integers $d \geq 1$ and $k \geq d+1$, a \textit{$(d, k)$-coloring} of a graph $G$ is a function
$\alpha\colon V(G) \to \{1, \dots, k\}$ such that for any pair of
distinct vertices $u$ and $v$, $\alpha(u) \neq \alpha(v)$ if $\dist_G(u,v) \leq d$.
In particular, a $(1, k)$-coloring of $G$ is also known as a \textit{proper $k$-coloring}.
If a graph $G$ has a $(d, k)$-coloring, we say that it is \textit{$(d, k)$-colorable}.
In this paper, we focus on the case $d \geq 2$.

One can generalize the concept of $(d, k)$-coloring to list $(d,k)$-coloring as follows.
A given function $L$ assigns to each vertex $v \in V(G)$ a list of possible colors $L(v) \subseteq \{1, \dots, k\}$.
A $(d, k)$-coloring $\alpha$ of $G$ is called a \textit{list $(d, k)$-coloring} if for every $v$, we have $\alpha(v) \in L(v)$.
In particular, if $L(v) = \{1, \dots, k\}$ for every $v \in V(G)$, then any list $(d, k)$-coloring of $G$ is also a $(d, k)$-coloring of
$G$ and vice versa.

\paragraph{(List) $(d, k)$-Coloring Reconfiguration}
Two (list) $(d, k)$-colorings $\alpha$ and $\beta$ of a graph $G$ are \textit{adjacent} if there exists exactly one $v \in V(G)$ such that
$\alpha(v) \neq \beta(v)$ and $\alpha(w) = \beta(w)$
for every $w \in V(G) - v$.
If $\beta$ is obtained from $\alpha$ (and vice versa) by recoloring only one $v$, we say that such a recoloring step is \textit{valid}.
Given two different (list) $(d, k)$-colorings $\alpha, \beta$ of a graph $G$, the \textsc{(List) $(d, k)$-\CR}  problem asks if there is a
sequence of (list) $(d, k)$-colorings
$\langle \alpha_0, \alpha_1, \dots, \alpha_\ell \rangle$ where $\alpha = \alpha_0$ and $\beta = \alpha_\ell$ such that $\alpha_i$ and
$\alpha_{i+1}$ are adjacent for every
$0 \leq i \leq \ell - 1$.
Such a sequence, if exists, is called a \textit{reconfiguration sequence} (i.e., a sequence of valid recoloring steps) between $\alpha$
and $\beta$.
An instance of \textsc{List $(d, k)$-\CR} is usually denoted by the $4$-tuple $(G, \alpha, \beta, L)$ and an instance of $(d, k)$-\CR by
the triple $(G, \alpha, \beta)$.

\section{\PSPACE-Completeness on Bipartite, Planar, and $2$-Degenerate Graphs}\label{sec:completeness}

In this section, we will prove \cref{thm:main}. Let us restate the theorem here.

\thmmain*

\subsection{Outline}

In \cref{sec:dk-col-reconf}, we describe a polynomial-time reduction from \textsc{List $(d, k)$-\CR} to $(d, k')$-\CR, where $k^\prime=\Theta(nd+k)$ and $n$ is the number of vertices in the input graph.
In \cref{sec:sliding-tokens}, we introduce a \PSPACE-complete variant of \ST.
In \cref{sec:list-dk-col-reconf}, we complete our reduction by describing a polynomial-time reduction from our variant to \textsc{List $(d, k)$-\CR}. 

\subsection{Reduction to $(d, k^\prime)$-\CR}\label{sec:dk-col-reconf}

{In this section, we present a polynomial-time reduction from \textsc{List $(d, k)$-\CR} to $(d, k^\prime)$-\CR, where $k^\prime = n(\lceil d/2 \rceil - 1) + 2 + k$ and $n = |V(G)|$. In particular, for every fixed integer $d\ge 2$, this implies that $(d,k)$-\CR is \PSPACE-complete when $k$ is part of the input.}
Specifically, we prove the following theorem.

\begin{theorem} \label{thm:coltolistcol}

For any fixed integer $d \geq 2$, given an instance $(G, \alpha, \beta, L)$ of \textsc{List $(d, k)$-\CR}, we can construct in polynomial time an instance $(G^\prime, \alpha^\prime, \beta^\prime)$ of $(d, k^\prime)$-\CR where $k^\prime = n(\lceil d/2 \rceil - 1) + 2 + k$ and $n = |V(G)|$, such that $(G, \alpha, \beta, L)$ is a yes-instance if and only if $(G^\prime, \alpha^\prime, \beta^\prime)$ is a yes-instance.

\end{theorem}
Note that any $(d, k)$-coloring of a graph $G$ is a list $(d, k)$-coloring of $G$ where $L(v) = \{1, \dots, k\}$ for every $v \in V(G)$.
To simulate the behavior of a list $(d, k)$-coloring, we need to constrain the available colors for each vertex $v$ to those in $L(v)$, which we achieve using \textit{frozen graphs}.
A graph $F$ with a $(d, k)$-coloring $\alpha$ is called a \textit{frozen graph} if no vertex in $F$ can be recolored---that is, there exists no reconfiguration sequence between $\alpha$ and any other $(d, k)$-coloring $\beta$ of $F$.
For each vertex $v$ in $G$, we construct a corresponding frozen graph $F_v$ and strategically position its vertices to enforce color restrictions: placing vertices of $F_v$ at distance $d+1$ from $v$ if their colors are in $L(v)$, and at distance at most $d$ otherwise.

\subsubsection{Frozen Graphs}
We begin by describing how a (precolored) frozen graph $F_v$ and its $(d, k^\prime)$-coloring $\alpha_v$ can be constructed for a vertex
$v \in V(G)$, where $k^\prime = n(\lceil d/2 \rceil - 1) + 2 + k$ and $n = |V(G)|$.
The gadget is illustrated in \cref{fig:listgadget}.
We emphasize that $v$ does not belong to its corresponding frozen graph $F_v$.

For $d \geq 3$, our construction is as follows.
First, for each $v \in V(G)$, we create a central vertex $c_v$.
We then construct a path $T_v$ which includes $c_v$ as an endpoint and has length $\lceil d/2 \rceil - 1$.
Suppose that $T_v = c_v c_{v,1}\dots c_{v,\lceil d/2 \rceil - 1}$.
Let $c^\prime_v = c_{v,\lceil d/2 \rceil - 1}$ be the endpoint of $T_v$ other than $c_v$.
Let $C_0 \notin \{1, \dots, k\}$ be a fixed color.
We color the vertices of $T_v$ starting from $c_v$ by using the color $C_0$ for $c_v$ and $\lceil d/2 \rceil - 1$ other distinct new
colors $C_{v, 1}, C_{v, 2}, \dots, C_{v, \lceil d/2 \rceil - 1}$ for the remaining vertices $c_{v,1}, \dots, c_{v,\lceil d/2 \rceil - 1}$,
respectively.
In particular, $c^\prime_v$ has color $C_{v, \lceil d/2 \rceil - 1}$.
We also remark that none of $C_{v, 1}, C_{v, 2}, \dots, C_{v, \lceil d/2 \rceil - 1}$ is in $\{1, \dots, k\}$.
At this point, so far, for each $v \in V(G)$, we have used $\lceil d/2 \rceil - 1$ distinct colors for vertices other than $c_v$ in each
$T_v$ and one fixed color $C_0$ for every $c_v$.
Thus, in total, $n(\lceil d/2 \rceil - 1) + 1$ distinct colors have been used.

Next, for each $v \in V(G)$ and each vertex $u \neq v$, we construct a (possibly trivial) path $T^v_u$ which includes $c_v$ as an endpoint and has
length $t := \lfloor d/2 \rfloor - 1$.
If $t = 0$ (equivalently, $d = 3$), then $T^v_u$ consists only of the vertex $c_v$ and we define ${c^\prime}^v_u := c_v$.
Otherwise, we write $T^v_u = c_v c^v_{u,1}\dots c^v_{u,t}$ and denote by ${c^\prime}^v_u = c^v_{u,t}$ the endpoint of $T^v_u$ other than $c_v$.
\begin{itemize}
    \item When $d$ is even, we have $\lceil d/2 \rceil - 1 = \lfloor d/2 \rfloor - 1$, i.e., the number of vertices in $T_u - c_u$ is
          equal to the number of vertices in $T^v_u - c_v$.
          In this case, we color the vertices of $T^v_u$ starting from $c_v$ by using the color $C_0$ for $c_v$ and the
          $\lceil d/2 \rceil - 1 = \lfloor d/2 \rfloor - 1$ other distinct colors $C_{u, 1}, C_{u, 2}, \dots, C_{u, \lfloor d/2 \rfloor - 1}$ respectively for the remaining vertices
          $c^v_{u,1}, \dots, c^v_{u, \lfloor d/2 \rfloor - 1}$.
          In particular, the endpoint ${c^\prime}^v_u$ has color $C_{u, \lfloor d/2 \rfloor - 1}$.
          (We note that all these colors are used to color vertices in $T_u$ for $u \in V(G)$.)

    \item When $d$ is odd, we have $\lceil d/2 \rceil - 1 = (\lfloor d/2 \rfloor - 1) + 1$, i.e., the number of vertices in $T_u - c_u$ is
          equal to the number of vertices in $T^v_u - c_v$ plus one.
          In this case, we color the vertices of $T^v_u$ starting from $c_v$ by using the color $C_0$ for $c_v$ and the
          $\lfloor d/2 \rfloor - 1$ other distinct colors $C_{u, 2}, C_{u, 3}, \dots, C_{u, \lceil d/2 \rceil - 1}$ respectively for the
          remaining vertices $c^v_{u,1}, \dots,\allowbreak c^v_{u, \lfloor d/2 \rfloor - 1}$, leaving one color $C_{u, 1}$ that has not yet been used.
          To handle this situation, we add a new vertex ${c^\star}^v_u$ adjacent to $c_v$ and color it by the color $C_{u, 1}$.
\end{itemize}

To finish our construction of $F_v$ and $\alpha_v$ for each $v \in V(G)$, we pick some vertex {$u_0$} in $G$ other than $v$.
Additionally, we add $k + 1$ extra new vertices labelled $c^\star_v, w_{v, 1}, \dots, w_{v, k}$.
We then join $c^\star_v$ to any $w_{v, i}$ where $i \in L(v) \subseteq \{1, \dots, k\}$
and join the endpoint {${c^\prime}^v_{u_0}$} of {$T^v_{u_0}$} to $c^\star_v$ and to any $w_{v, i}$ where $i \notin L(v)$.
Let $C_1$ be a fixed color that is different from any colors that have been used.
We finally color $c^\star_v$ by $C_1$, and each $w_{v, i}$ by the color $i \in \{1, \dots, k\}$.
At this point, $k + 1$ extra distinct colors are used.
In total, we use $k^\prime = (n(\lceil d/2 \rceil - 1) + 1) + (k + 1) = n(\lceil d/2 \rceil - 1) + 2 + k$ colors.
This concludes our construction for $d \geq 3$.

For $d = 2$, we construct $F_v$ and $\alpha_v$ as follows. We create a central vertex $c_v$ and color it by a fixed color $C_0$.
We create a new vertex $c^\star_v$ adjacent to $c_v$ and color it by a fixed color $C_1$.
We then add $k$ extra vertices $w_{v, 1}, \dots, w_{v, k}$ and color them by the colors $1, \dots, k$ respectively.
Next, we join $c^\star_v$ to any $w_{v, i}$ ($1 \leq i \leq k$) and join $c_v$ to any $w_{v, i}$ where $i \notin L(v)$.
In total, we use $k + 2$ colors, which is equal to $n(\lceil d/2 \rceil - 1) + 2 + k$ when $d = 2$.

\begin{lemma}\label{lem:frozengraph}
    Our construction correctly produces a frozen graph $F_v$ with its $(d, k^\prime)$-coloring $\alpha_v$.
\end{lemma}
\begin{proof}
    {For completeness, we verify the case $d=2$. In this case, $V(F_v)=\{c_v,{c^\star_v},w_{v,1},\dots,w_{v,k}\}$ with edges $c_v{c^\star_v}$, all edges ${c^\star_v} w_{v,i}$, and additionally edges $c_v w_{v,i}$ for all $i\notin L(v)$. Hence $\mathrm{diam}(F_v)\le 2=d$. Moreover, $\alpha_v$ is rainbow on $V(F_v)$ (with $\alpha_v(c_v)=C_0$, $\alpha_v({c^\star_v})=C_1$, and $\alpha_v(w_{v,i})=i$), and therefore no vertex is recolorable in a $(2,k^\prime)$-coloring; thus $(F_v,\alpha_v)$ is frozen for $d=2$.}
    {Hence, we may assume $d\ge 3$ for the remainder of the proof.}

    Let $s:=\lceil d/2\rceil-1$ and $t:=\lfloor d/2\rfloor-1$, so that $s+t=d-2$.
    By construction, $F_v$ is a tree and $\alpha_v$ is rainbow on $V(F_v)$.

    We first show that $\diam(F_v)\leq d$.
        Since $F_v$ is a tree, it suffices to upper bound the distance between any two leaves.

        If $t=0$ (equivalently, $d=3$), then for every $u\neq v$ we have $T_u^v=\{c_v\}$ and ${c'}^{v}_{u}=c_v$.
        In this case, every leaf of $F_v$ is either the endpoint $c_v'$ of $T_v$, a vertex $w_{v,i}$, or (since $d$ is odd) a vertex ${{c^\star}^v_u}$ adjacent to $c_v$.
        Moreover, the unique leaves at distance $2$ from $c_v$ are the vertices $w_{v,i}$ with $i\in L(v)$ (which are adjacent to ${c^\star_v}$), while all other leaves are adjacent to $c_v$.
        Hence any two leaves are at distance at most $3$, and thus $\diam(F_v)\leq 3=d$.

        Assume now that $t\geq 1$.
        Let $u_0\in V(G)\setminus\{v\}$ be the selected vertex whose endpoint ${c'}^{v}_{u_0}$ is adjacent to ${c^\star_v}$ and to all vertices $w_{v,i}$ with $i\notin L(v)$.
        Every leaf of $F_v$ belongs to one of the following types:
        (i) the endpoint $c_v'$ of $T_v$,
        (ii) an endpoint ${c'}^{v}_{u}$ of $T_u^v$ with $u\in V(G)\setminus\{v,u_0\}$,
        (iii) a vertex $w_{v,i}$,
        and, when $d$ is odd, (iv) a vertex ${{c^\star}^v_u}$ with $u\in V(G)\setminus\{v\}$.
        A direct distance bound between these leaf types yields $\diam(F_v)\leq d$:
        \begin{itemize}
            \item $\dist_{F_v}(c_v', {c'}^{v}_{u}) = s+t = d-2$ for every $u\neq v$.
            \item For any two distinct vertices $u,u'\in V(G)\setminus\{v\}$, we have $\dist_{F_v}({c'}^{v}_{u}, {c'}^{v}_{u'}) = 2t \le d$.
            \item For every $i\notin L(v)$, we have $\dist_{F_v}(c_v', w_{v,i}) = (s+t)+1 = d-1$.
            \item For every $i\in L(v)$, we have $\dist_{F_v}(c_v', w_{v,i}) = (s+t)+2 = d$ (via ${c'}^{v}_{u_0}$ and ${c^\star_v}$).
            \item For every $u\in V(G)\setminus\{v,u_0\}$ and $i\in[k]$, we have $\dist_{F_v}({c'}^{v}_{u}, w_{v,i}) \leq 2t+2\leq d$ (and in fact $\leq d-1$ when $d$ is odd).
            \item When $d$ is odd, for every $u\in V(G)\setminus\{v\}$ and $i\in[k]$ we have
            \begin{align*}
            \dist_{F_v}(c_v',{{c^\star}^v_u}) & = s+1 \le d, \\
            \dist_{F_v}({c'}^{v}_{u},{{c^\star}^v_{u'}}) & \le t+1 \le d, \\
            \dist_{F_v}({{c^\star}^v_u},w_{v,i}) & \le t+3 \le d, \\
            \dist_{F_v}({{c^\star}^v_u},{{c^\star}^v_{u'}}) & = 2 \quad \text{for all } u \neq u'.
            \end{align*}
            \item For any $i,j\in[k]$, we have $\dist_{F_v}(w_{v,i},w_{v,j})\le 3$.
        \end{itemize}

    Now fix any vertex $x\in V(F_v)$.
    Since $\alpha_v$ is rainbow, every color $c\neq \alpha_v(x)$ appears on a unique vertex $y\in V(F_v)$.
    As $\diam(F_v)\leq d$, we have $\dist_{F_v}(x,y)\leq d$, and therefore $x$ cannot be recolored to $c$ in a $(d,k^\prime)$-coloring.
    Hence no vertex of $F_v$ is recolorable.
\end{proof}
One can verify that our construction indeed can be done in polynomial time.

\begin{figure}[ht]
    \centering
    \colorlet{newcolors}{Mulberry}
    \colorlet{shortpath}{blue}
    \colorlet{longpath}{red}
    \colorlet{originalc}{GreenYellow}
    \begin{tikzpicture}[scale=0.7, transform shape]
        \tikzmath{\x = 5; \y =\x - 1;}        \tikzstyle{basic}=[draw,thick,circle,minimum size=8mm,inner sep=0pt]

        \node[basic,newcolors] (c) at (0,0) {\LARGE $c_v$};
        \node[basic,shortpath] (1) at (360/\x: 3cm) {\Large ${c^\prime}^v_{u_1}$};
        \draw[shortpath] (c) edge["\Large $T^v_{u_1}$"] (1);

        \foreach \i in {2,3,...,\y} {
        \node[basic,shortpath] (\i) at (\i*380/\x: 3cm) {\Large ${c^\prime}^v_{u_{\i}}$};        \draw[shortpath] (c) edge["\Large $T^v_{u_\i}$"] (\i);
        }
        \node[basic,originalc,text=black] (v) at (4.5,-1) {\LARGE $v$};
        \draw[longpath] (c) edge["\Large $T_v+v$"] (v);

        \node[basic] (u) at ([shift={(1,-1)}]v.center) {\LARGE $u$};
        \draw (v) -- (u);
        \draw[dashed] (u) -- ([shift={(1.25,-0.25)}]u.center);
        \draw[dashed] (u) -- ([shift={(1.25,-0.75)}]u.center);
        \draw[dashed] (u) -- ([shift={(1.25,-1.25)}]u.center);

        \node[basic,shortpath] (lastend) at (0,-2.5) {\Large ${c^\prime}^v_{u_5}$};
        \node[basic,newcolors] (last) at (0,-4) {\Large ${c^\star_v}$};
        \draw (lastend) -- (last);
        \draw[shortpath] (c) edge["\Large $T^v_{u_5}$"] (lastend);

        \node[basic] (1) at (-3,-5.25) {\Large $w_{v, 1}$};
        \node[basic] (2) at (-1,-5.25) {\Large $w_{v, 2}$};
        \node[basic] (3) at (1, -5.25) {\Large $w_{v, 3}$};
        \node[basic] (4) at (3, -5.25) {\Large $w_{v, 4}$};
        \node[basic] (5) at (5, -5.25) {\Large $w_{v, 5}$};

        \draw (lastend) -- (1);
        \draw (last) -- (2);
        \draw (last) -- (3);
        \draw[bend right] (lastend) -- (4);
        \draw (lastend) -- (5);

        \begin{pgflowlevelscope}{\pgftransformscale{0.9}}
            \matrix [draw,below left] at ($(current bounding box.north east)+(0.8,0.8)$) {
            \node [fill=shortpath,label=right:$\lfloor \frac{d}{2}\rfloor-1$ Edges] {};             \\
            \node [fill=longpath,label=right:$\lceil \frac{d}{2}\rceil $ Edges] {};                \\
            \node [fill=newcolors,label=right:{Colors $C_0,C_1$}] {};                           \\
            };
        \end{pgflowlevelscope}
    \end{tikzpicture}
    \caption{An example of a vertex $v$ joining to its corresponding frozen graph $F_v$. Here $d$ is even, $k = 5$, $G$ is some list $(d,k)$-colorable graph having six vertices labelled $v, u_1, u_2, \dots, u_5$ (note that here we selected $u_0 = u_5$), and $L(v) = \{2, 3\} \subseteq \{1, \dots, 5\}$.\label{fig:listgadget}}
\end{figure}

\subsubsection{Construction of An Instance $(G^\prime, \alpha^\prime, \beta^\prime)$ of $(d, k^\prime)$-\CR}

Given an instance $(G, \alpha, \beta, L)$ of \textsc{List $(d, k)$-\CR}, we now describe how to construct $G^\prime$ and its two
$(d, k^\prime)$-colorings $\alpha^\prime, \beta^\prime$, where $k^\prime = n(\lceil d/2 \rceil - 1) + 2 + k$.
To construct $G^\prime$, we start from the original graph $G$,
and we construct $F_v$ for each $v \in V(G)$ as described before.
Then, for $d \geq 3$, we simply join $v$ to $c^\prime_v$ --- the endpoint of $T_v$ other than $c_v$.
For $d = 2$, we join $v$ to $c_v$.
To construct $\alpha^\prime$ from $\alpha$, we simply assign $\alpha^\prime(v) = \alpha(v)$ for any $v \in V(G)$ and
$\alpha^\prime(w) = \alpha_v(w)$ for any $w \in V(F_v)$.
The construction of $\beta^\prime$ is similar.
One can verify that for any $v \in V(G)$, no vertex in $F_v$ can be recolored in $G^\prime$.
Again, to see this, observe that $(F_v,\alpha_v)$ is frozen by \cref{lem:frozengraph}.
Since $F_v$ is an induced subgraph of $G'$, adding vertices outside $F_v$ can only introduce additional distance-$d$ constraints,
and therefore no vertex of $F_v$ becomes recolorable in $G'$.
One can also verify that our construction can be done in polynomial time.

The following simple observation will be useful later in our proof.

\begin{lemma}\label{lem:dkcol-graph-type}
    Suppose that $G$ is planar, bipartite, and $2$-degenerate.
    The constructed graph $G^\prime$ is planar, bipartite, and $2$-degenerate if $d \geq 3$, and is planar and $2$-degenerate but not necessarily bipartite if $d = 2$.
\end{lemma}
\begin{proof}
    For $d \geq 3$, each frozen graph is planar, bipartite, and $1$-degenerate (as it is essentially a tree). Thus, from our construction of $G^\prime$, since $G$ is planar, bipartite, and $2$-degenerate, so is $G^\prime$.

    On the other hand, for $d = 2$, each frozen graph is planar and $2$-degenerate. Thus, from our construction of $G^\prime$, since $G$ is planar and $2$-degenerate, so is $G^\prime$.
    When $L(v)$ is a proper subset of $\{1, \dots, k\}$ for some $v \in V(G)$, the frozen graph $F_v$ contains a triangle formed by $c_v$, $c^\star_v$, and any $w_{v, i}$ where $i \notin L(v)$. In this case, $G^\prime$ is not bipartite.
\end{proof}

In the following lemma, we show that our construction correctly produces an instance of \textsc{$(d, k^\prime)$-\CR}.

\begin{lemma}\label{lem:dkcol-inst}
    $\alpha^\prime$ is a $(d, k^\prime)$-coloring of $G^\prime$.
    Consequently, so is $\beta^\prime$.
\end{lemma}

\begin{proof}
To show that $\alpha^\prime$ is a $(d,k^\prime)$-coloring of $G^\prime$, we prove that
\[
(\star)\qquad \text{for all distinct }x,y\in V(G^\prime),\ \alpha^\prime(x)=\alpha^\prime(y)\ \Rightarrow\ \dist_{G^\prime}(x,y)>d.
\]

For each $v\in V(G)$, the gadget $F_v$ is attached to $G$ by a single edge, namely $vc_v$ if $d=2$, and $vc_v^\prime$ if $d\ge 3$.
Hence:
\begin{enumerate}[label=(P\arabic*)]
    \item\label{it:P1} For all $a,b\in V(G)$, $\dist_{G^\prime}(a,b)=\dist_G(a,b)$.
    \item\label{it:P2} For all $a\in V(G)$ and all $y\in V(F_v)$,
    \[
        \dist_{G^\prime}(a,y)=\dist_G(a,v)+\dist_{G^\prime}(v,y).
    \]
    \item\label{it:P3} For all $x\in V(F_u)$ and $y\in V(F_v)$ with $u\neq v$,
    \[
        \dist_{G^\prime}(x,y)=\dist_{G^\prime}(x,u)+\dist_G(u,v)+\dist_{G^\prime}(v,y).
    \]
\end{enumerate}
Indeed, any path entering $F_v$ from outside must pass through its anchor $v$, and any path between two distinct gadgets must pass through both anchors and a path in $G$ between them.

\begin{itemize}
    \item \textbf{Case 1: $d=2$.} 
    
    Let $x\neq y$ with $\alpha^\prime(x)=\alpha^\prime(y)$.
    If $x,y\in V(G)$, then \ref{it:P1} and the fact that $\alpha$ is a $(2,k)$-coloring give $\dist_{G^\prime}(x,y)=\dist_G(x,y)>2$.
    If $x\in V(G)$ and $y\in V(F_v)$, then $\alpha^\prime(x)\in\{1,\dots,k\}$, hence $y=w_{v,i}$ for $i:=\alpha^\prime(x)$.
    If $x=v$, then $i=\alpha(v)\in L(v)$ and thus $c_vw_{v,i}\notin E(G^\prime)$, so $\dist_{G^\prime}(v,w_{v,i})=3>2$.
    If $x\neq v$, then by \ref{it:P2} we have
    \[
    \dist_{G^\prime}(x,w_{v,i})=\dist_G(x,v)+\dist_{G^\prime}(v,w_{v,i})\ge 1+2=3>2,
    \]
    since always $\dist_{G^\prime}(v,w_{v,i})\ge 2$.
    Finally, if $x\in V(F_u)$ and $y\in V(F_v)$ with $u\neq v$, we distinguish the common color.
    If it is $i\in\{1,\dots,k\}$, then $x=w_{u,i}$ and $y=w_{v,i}$, and \ref{it:P3} yields
    $\dist_{G^\prime}(x,y)\ge 2+1+2=5>2$.
    If it is $C_0$, then $x=c_u$ and $y=c_v$, and $\dist_{G^\prime}(x,y)\ge 1+1+1=3>2$.
    If it is $C_1$, then $x=c_u^\star$ and $y=c^\star_v$, and $\dist_{G^\prime}(x,y)\ge 2+1+2=5>2$.
    Thus $(\star)$ holds for $d=2$.
    
    \item \textbf{Case 2: $d\ge 3$.}
    
    Let $x\neq y$ with $\alpha^\prime(x)=\alpha^\prime(y)$.
    If $x,y\in V(G)$, then \ref{it:P1} and the fact that $\alpha$ is a $(d,k)$-coloring give $\dist_{G^\prime}(x,y)=\dist_G(x,y)>d$.
    If $x,y\in V(F_v)$ for some $v$, then $\alpha_v$ is rainbow on $V(F_v)$, so no two distinct vertices of $F_v$ share a color; thus this case cannot occur.
    Hence, it remains to consider the cases where $x$ and $y$ lie in different parts.

    \begin{itemize}
        \item \textbf{Case 2.1: $x\in V(G)$ and $y\in V(F_v)$.}
        
        Then $\alpha^\prime(x)=\alpha(x)=i\in\{1,\dots,k\}$.
        By construction, the unique vertex of $F_v$ colored $i$ is $w_{v,i}$, hence $y=w_{v,i}$.
        If $x=v$, then $i=\alpha(v)\in L(v)$ and (letting $u_0\neq v$ be the vertex selected when constructing $F_v$) we have $\dist_{G^\prime}(v,w_{v,i})
        =
        \dist_{G^\prime}(v,c_v)+\dist_{F_v}(c_v,{c^\prime}^{v}_{u_0})+\dist_{G^\prime}({c^\prime}^{v}_{u_0},c^\star_v)+\dist_{G^\prime}(c^\star_v,w_{v,i})
        =
        \lceil d/2\rceil+(\lfloor d/2\rfloor-1)+1+1
        =d+1>d$,
        since $w_{v,i}$ is adjacent to $c^\star_v$, and $c^\star_v$ is adjacent to ${c^\prime}^{v}_{u_0}$.
        If $x\neq v$, then by \ref{it:P2} we have $\dist_{G^\prime}(x,w_{v,i})=\dist_G(x,v)+\dist_{G^\prime}(v,w_{v,i})\ge 1+d>d$,
        because always $\dist_{G^\prime}(v,w_{v,i})\ge d$.

        \item \textbf{Case 2.2: $x\in V(F_u)$ and $y\in V(F_v)$ with $u\neq v$.}
        
        We distinguish the common color $c:=\alpha^\prime(x)=\alpha^\prime(y)$.
        \begin{itemize}
            \item[(a)] $c=i\in\{1,\dots,k\}$.
            
            Then $x=w_{u,i}$ and $y=w_{v,i}$.
            Using \ref{it:P3} and $\dist_G(u,v)\ge 1$, we have $\dist_{G^\prime}(x,y)
            =
            \dist_{G^\prime}(u,w_{u,i})+\dist_G(u,v)+\dist_{G^\prime}(v,w_{v,i})
            \ge d+1+d>d$,
            since $\dist_{G^\prime}(u,w_{u,i})\ge d$ and $\dist_{G^\prime}(v,w_{v,i})\ge d$ for all $i$.

            \item[(b)] $c=C_0$.
            
            Then $x=c_u$ and $y=c_v$.
            By construction $\dist_{G^\prime}(u,c_u)=\lceil d/2\rceil$ and $\dist_{G^\prime}(v,c_v)=\lceil d/2\rceil$, hence by \ref{it:P3}, we have $\dist_{G^\prime}(x,y)=\lceil d/2\rceil+\dist_G(u,v)+\lceil d/2\rceil
            \ge 2\lceil d/2\rceil+1>d$.
            
            \item[(c)] $c=C_1$.
            
            Then $x=c_u^\star$ and $y=c^\star_v$.
            By construction $\dist_{G^\prime}(u,c_u^\star)=d$ and $\dist_{G^\prime}(v,c^\star_v)=d$, so \ref{it:P3} yields $\dist_{G^\prime}(x,y)=d+\dist_G(u,v)+d\ge 2d+1>d$.

            \item[(d)] $c=C_{p,j}$ for some $p\in V(G)$ and $1\le j\le \lceil d/2\rceil-1$.
            
            We use the following distances from the construction.
            In $F_p$, the unique vertex of color $C_{p,j}$ lies on $T_p$ at distance $j$ from $c_p$, hence
            \begin{equation}\label{eq:dist-in-Fp}
            \dist_{G^\prime}(p,z)=\lceil d/2\rceil-j \quad\text{for }z\in V(F_p)\text{ with }\alpha^\prime(z)=C_{p,j}.
            \end{equation}
            For $q\neq p$, the vertices of color $C_{p,j}$ in $F_q$ lie on the $p$-branch of $F_q$; thus
            \begin{equation}\label{eq:dist-in-Fq}
            \dist_{G^\prime}(q,z)=
            \begin{cases}
            \lceil d/2\rceil+j, & \text{if $d$ is even,}\\
            \lceil d/2\rceil+1, & \text{if $d$ is odd and $j=1$,}\\
            \lceil d/2\rceil+(j-1), & \text{if $d$ is odd and $j\ge 2$,}
            \end{cases}
            \end{equation}
            for $z \in V(F_q)$ with $\alpha^\prime(z)=C_{p,j}$.

            If $p\in\{u,v\}$, assume without loss of generality that $v=p$.
            Then \ref{it:P3}, \eqref{eq:dist-in-Fp}, \eqref{eq:dist-in-Fq}, and $\dist_G(u,p)\ge 1$ imply:
            \begin{align*}
            &\dist_{G^\prime}(x,y) = \dist_{G^\prime}(x,u)+\dist_G(u,p)+\dist_{G^\prime}(p,y)\\
            &\ge \begin{cases}
            (\lceil d/2\rceil+j)+1+(\lceil d/2\rceil-j)=d+1, & \text{if $d$ is even,}\\
            (\lceil d/2\rceil+1)+1+(\lceil d/2\rceil-1)=2\lceil d/2\rceil+1=d+2, & \text{if $d$ is odd and $j=1$,}\\
            (\lceil d/2\rceil+j-1)+1+(\lceil d/2\rceil-j)=2\lceil d/2\rceil=d+1, & \text{if $d$ is odd and $j\ge 2$,}
            \end{cases}
            \end{align*}     
            and in all cases $\dist_{G^\prime}(x,y)>d$.

            If $p\notin\{u,v\}$, then \ref{it:P3}, \eqref{eq:dist-in-Fq}, and $\dist_G(u,v)\ge 1$ give
            \begin{align*}
                &\dist_{G^\prime}(x,y) = \dist_{G^\prime}(x,u)+\dist_G(u,v)+\dist_{G^\prime}(v,y) \\
                &\ge
                \begin{cases}
                (\lceil d/2\rceil+j)+1+(\lceil d/2\rceil+j)=2\lceil d/2\rceil+2j+1>d, & \text{if $d$ is even,}\\
                (\lceil d/2\rceil+1)+1+(\lceil d/2\rceil+1)=2\lceil d/2\rceil+3>d, & \text{if $d$ is odd and $j=1$,}\\
                (\lceil d/2\rceil+j-1)+1+(\lceil d/2\rceil+j-1)=2\lceil d/2\rceil+2j-1>d, & \text{if $d$ is odd and $j\ge 2$.}
                \end{cases}
            \end{align*}
            and again $\dist_{G^\prime}(x,y)>d$.
        \end{itemize}
        This completes \textbf{Case~2.2} and hence establishes $(\star)$ for $d\ge 3$.
    \end{itemize}

\end{itemize}

Finally, $\beta^\prime$ is constructed from $\beta$ in the same way that $\alpha^\prime$ is constructed from $\alpha$, so the same argument shows that $\beta^\prime$ is also a $(d,k^\prime)$-coloring of $G^\prime$.
\end{proof}

\subsubsection{The Correctness of Our Reduction}

    We are now ready to prove the correctness of our reduction which will prove \cref{thm:coltolistcol}.

\begin{lemma}\label{lem:dkcol-correct}
    $(G, \alpha, \beta, L)$ is a yes-instance if and only if $(G^\prime, \alpha^\prime, \beta^\prime)$ is a yes-instance.
\end{lemma}
\begin{proof}
    {
    By \cref{lem:frozengraph}, for every $v\in V(G)$ the graph $F_v$ satisfies $\diam(F_v)\le d$, and by construction $\alpha_v$ is a coloring of $F_v$ that uses all $k^\prime$ colors exactly once.
    Fix any vertex $x\in V(F_v)$ and any color $c\in\{1,\dots,k^\prime\}\setminus\{\alpha_v(x)\}$.
    Since $\alpha_v$ uses all $k^\prime$ colors exactly once, there exists a unique vertex $y\in V(F_v)$ with $\alpha_v(y)=c$, and since $\diam(F_v)\le d$ we have $\dist_{F_v}(x,y)\le d$.
    Therefore recoloring $x$ to $c$ would immediately violate the $(d,k^\prime)$-constraint (already within $F_v$), and hence no vertex of $F_v$ can be recolored in any valid recoloring sequence of $G^\prime$.}

    Fix $v\in V(G)$.
    For every $i\in\{1,\dots,k\}\setminus L(v)$, the vertex $w_{v,i}\in V(F_v)$ has color $i$ and satisfies $\dist_{G^\prime}(v,w_{v,i})=d$ by construction.
    Hence $v$ can never be recolored to any color in $\{1,\dots,k\}\setminus L(v)$.

    Moreover, $v$ can never be recolored to any of the auxiliary colors in $\{C_0,C_1\}\cup\{C_{u,j}\mid u\in V(G),\,1\le j\le \lceil d/2\rceil-1\}$:
    the gadget $F_v$ contains a frozen vertex of each such color within distance at most $d$ from $v$.
    Indeed, $\dist_{G^\prime}(v,c_v)=1+(\lceil d/2\rceil-1)=\lceil d/2\rceil$ and $\alpha_v(c_v)=C_0$,
    while $\dist_{G^\prime}(v,c^\star_v)=d$ and $\alpha_v(c^\star_v)=C_1$.
    Finally, for every auxiliary color $C_{u,j}$, let $x_{u,j}^v$ be the unique vertex of $F_v$ colored by $C_{u,j}$; by construction $\dist_{F_v}(c_v,x_{u,j}^v)\le \lceil d/2\rceil-1$, and thus
    \[
        \dist_{G^\prime}(v,x_{u,j}^v)\le \dist_{G^\prime}(v,c_v)+\dist_{F_v}(c_v,x_{u,j}^v)\le \lceil d/2\rceil+(\lceil d/2\rceil-1)\le d.
    \]
    Since $x_{u,j}^v$ is frozen, $v$ is blocked from using $C_{u,j}$.

    Consequently, along any recoloring sequence in $G^\prime$, every vertex $v\in V(G)$ always has a color in $\{1,\dots,k\}$ and this color always lies in its list $L(v)$.
    Therefore, projecting any recoloring sequence in $G^\prime$ onto $V(G)$ yields a valid sequence for the list-instance $(G,\alpha,\beta,L)$.
    To see that this extension is valid in $G^\prime$, consider any recoloring step of the list-sequence that changes the color of a vertex $v\in V(G)$ to some $i\in L(v)$.
    Since the step is valid in $(G,\alpha,\beta,L)$, no vertex of $G$ within distance at most $d$ from $v$ has color $i$ at that moment.
    In $G^\prime$, the only additional vertices of color $i$ are the gadget vertices $w_{u,i}$ with $u\in V(G)$.
    If $u=v$, then by construction $\dist_{G^\prime}(v,w_{v,i})=d+1$, so $w_{v,i}$ does not block the recoloring of $v$.
{Suppose first that $d=2$. If $u\neq v$, then the only edge between $F_u$ and $G$ is $u c_u$.
Since $w_{u,i}\neq c_u$, we have $\dist_{F_u}(c_u,w_{u,i})\ge 1$, and thus
\[
    \dist_{G^\prime}(v,w_{u,i})=\dist_{G}(v,u)+1+\dist_{F_u}(c_u,w_{u,i})\ge \dist_G(v,u)+2\ge 3=d+1.
\]
Hence no vertex $w_{u,i}$ blocks the recoloring of $v$. We may therefore assume $d\ge 3$ in the following.}

    If $u\neq v$, then the only edge between $F_u$ and $G$ is $u c_u^\prime$, so any $vw_{u,i}$-path in $G^\prime$ must pass through $u$ and then enter $F_u$ via $c_u^\prime$.
    Hence
    \[
        \dist_{G^\prime}(v,w_{u,i})=\dist_{G}(v,u)+1+\dist_{F_u}(c_u^\prime,w_{u,i}).
    \]
    Moreover, $\dist_{F_u}(c_u^\prime,w_{u,i})\ge d-1$ holds for all $i\in[k]$:
    indeed, if $i\notin L(u)$ then $w_{u,i}$ is adjacent to the distinguished endpoint ${c'}_{u_0}^{u}$, and $\dist_{F_u}(c_u^\prime,{c'}_{u_0}^{u})=d-2$;
    if $i\in L(u)$ then $w_{u,i}$ is adjacent to $c_u^\star$, which is adjacent to ${c'}_{u_0}^{u}$, yielding $\dist_{F_u}(c_u^\prime,w_{u,i})=d$.
    Since $\dist_G(v,u)\ge 1$ for $u\neq v$, we obtain $\dist_{G^\prime}(v,w_{u,i})\ge d+1$.
    Therefore no gadget vertex of color $i$ lies within distance $d$ of $v$, and the recoloring step remains valid in $G^\prime$.

    Conversely, any valid list-recoloring sequence of $(G,\alpha,\beta,L)$ extends to $G^\prime$ by keeping all vertices of each $F_v$ fixed.
\end{proof}

To conclude this section, we show that the palette size $k^\prime$ in our reduction is asymptotically tight up to constant factors.
\begin{lemma}\label{lem:dkcol-color-lb}
{Assume $d\geq 3$ and let $F_v$ be the gadget constructed above for some $v\in V(G)$.
Then $\diam(F_v)\le d$, and hence every two distinct vertices of $F_v$ are at distance at most $d$.
Consequently, in any $(d,\ell)$-coloring of $F_v$ all vertices must receive pairwise distinct colors, and therefore $\ell\ge |V(F_v)| = k^\prime$.
In particular, within the present list-to-nonlist reduction scheme, the palette size $k^\prime = k+2+n(\lceil d/2\rceil-1)=\Theta(nd+k)$ in \cref{thm:coltolistcol} is asymptotically tight up to constant factors.}
\end{lemma}
\begin{proof}
{By \cref{lem:frozengraph} we have $\diam(F_v)\le d$.
Thus for any two distinct vertices $x,y\in V(F_v)$ we have $\dist_{F_v}(x,y)\le d$, and a $(d,\ell)$-coloring cannot assign the same color to both $x$ and $y$.
Hence at least $|V(F_v)|$ colors are required on $F_v$.
Finally, our construction colors $F_v$ in a rainbow fashion using all $k^\prime$ colors, and therefore $|V(F_v)|=k^\prime=k+2+n(\lceil d/2\rceil-1)$.}
\end{proof}

\subsection{\textsc{Sliding Tokens}}\label{sec:sliding-tokens}

In this section, we first revisit a variant of \ST used by Bonsma and Cereceda~\cite{BonsmaC09} and then describe and prove
\PSPACE-completeness of our restricted variant. In particular, this will help prove the following theorem.

\begin{restatable}{theorem}{thmrstpspacec}\label{thm:rst-pspacec}
    \ST is \PSPACE-complete on graphs that are planar and $2$-degenerate.
\end{restatable}

\subsubsection{Bonsma and Cereceda's \textsc{Sliding Tokens} Variant}
In a graph $G$, a \textit{valid token configuration} is a set of vertices on which tokens are placed such that no two tokens are either
on the
same or adjacent vertices, i.e., each token configuration forms an \textit{independent set} of $G$.
A \textit{move} (or \textit{$\mathsf{TS}$-move}) between two token configurations of $G$ involves sliding a single
token from one vertex to one
of its (unoccupied) neighbors.
A move must always result in a valid token configuration.
Given a graph $G$ and two valid token configurations $T_A, T_B$, the \ST problem, first introduced by Hearn and Demaine~\cite{HearnD05}, asks if there is a
sequence of moves transforming $T_A$ into $T_B$.
Such a sequence, if it exists, is called a \textit{$\mathsf{TS}$-sequence} in $G$ between $T_A$ and $T_B$.

Bonsma and Cereceda~\cite{BonsmaC09} show that \ST is \PSPACE-complete even when restricted to the following set of $(G, T_A, T_B)$
instances.
For a more detailed explanation, we refer readers to the PhD thesis of Cereceda~\cite{Cereceda2007}.

\begin{itemize}
    \item The graph $G$ has three types of gadgets: \textit{token triangles} (a copy of $K_3$), \textit{token edges} (a copy of $K_2$), and
          \textit{link edges} (a copy of $K_2$).
          Token triangles and token edges are all mutually disjoint.
          They are joined together by link edges in such a way that every vertex of $G$ belongs to exactly one token triangle or one token
          edge.
          Moreover, every vertex in a token triangle has degree $3$, and $G$ has a planar embedding where every token triangle bounds a
          face. The graph $G$ has maximum degree $3$ and minimum degree $2$.
    \item The token configurations $T_A$ and $T_B$ are such that every token triangle and every token edge contains exactly one token on
          one of their vertices.
\end{itemize}

Valid token configurations where every token triangle and every token edge contains exactly one token on one of their vertices are called
\textit{standard token configurations}.
Thus, both $T_A$ and $T_B$ are standard.
One can verify that in any $\mathsf{TS}$-sequence in $G$ starting from $T_A$ or $T_B$, no token ever leaves its corresponding token
triangle/edge.

We define the \textit{degree} of a gadget as the number of gadgets of other types sharing exactly one common vertex with it.
By definition, a token triangle in $G$ has degree exactly $3$ because there are exactly three link edges, each of which shares a common vertex with it.
Each token edge has degree between $2$ and $4$ because any endpoint of the token edge has at most two link edges incident to that endpoint.
Two link edges may share a common vertex.
However, when calculating the degree of a link edge, we only count the number of token triangles/token edges sharing
exactly one common vertex with it and ignore any other link edge having the same property.
Since all token triangles and token edges are mutually disjoint, a link edge always has degree exactly $2$.

\subsubsection{Our \ST Variant}

In our \ST variant, we modify each instance in the above set using the following rules.

\begin{enumerate}[label=(R\arabic*)]
    \item For a single token edge of degree $4$, replace that token edge by two new token edges joined together by a single link edge.\label{item:rule:deg4}
    \item For a single link edge joining vertices of two degree $3$ gadgets, replace that link edge by two new link edges joined together
          by a single token edge.\label{item:rule:2deg3}
\end{enumerate}

We perform~\ref{item:rule:deg4} and then~\ref{item:rule:2deg3} sequentially: First we apply~\ref{item:rule:deg4} on the original graph repeatedly until no token edge of degree $4$ exists
in the resulting graph. We then continue by
applying~\ref{item:rule:2deg3} repeatedly until no link edge joining vertices of two degree $3$ gadgets remain in the resulting graph.
In each case, new tokens are appropriately added to ensure that the resulting token configuration is standard.
Additionally, if a vertex in the original graph does (not) have a token on it, then in the newly constructed graph, it does (not) too.

Let's call the final new corresponding instance $(G^\prime, T_A^\prime, T_B^\prime)$.
We note that after these modifications, each token triangle has degree exactly $3$ and each token edge has degree either $2$ or $3$. Moreover,
each token triangle or token edge of degree $3$ has a link edge to at least one token edge of degree $2$.
As $G$ is planar, the graph $G^\prime$ is planar too.
Additionally, from the modification, $G^\prime$ has maximum degree $3$ and minimum degree $2$, and both $T_A^\prime$ and $T_B^\prime$ are
standard token configurations.
One can readily verify that in any $\mathsf{TS}$-sequence in $G^\prime$ starting from $T_A^\prime$ or $T_B^\prime$, no token ever leaves
its corresponding token triangle/edge.

    \begin{figure}[!ht]
        \centering
        \begin{tikzpicture}[scale=1.0, every node/.style={circle, thick, draw, minimum size=5mm, transform shape}, token/.style={circle, thick, draw, minimum size=3mm, fill=black, transform shape}]
            \begin{scope}[shift={(0,0)}]
                \foreach \n/\x/\y in {T11/0/0, T12/2/0}
                    {
                        \node (\n) at (\x, \y) {};
                    }
                \draw[ultra thick] (T11) -- (T12);
                \draw (T11) -- ([shift={(-1,0.5)}] T11.center) (T11) -- ([shift={(-1,-0.5)}] T11.center) (T12) -- ([shift={(1,0.5)}] T12.center) (T12) -- ([shift={(1,-0.5)}] T12.center);
                \node [label=below:$u$] at (T11) {};
                \node [label=below:$v$] at (T12) {};

                \node[rectangle, text width=4cm, align=center, draw=none] at (1, -1.5) {Before};

                \node[token] at (T11) {};
            \end{scope}
            \begin{scope}[shift={(6,0)}]
                \foreach \n/\x/\y in {T11/0/0, T12/2/0, N1/0/1, N2/2/1}
                    {
                        \node (\n) at (\x, \y) {};
                    }
                \draw[ultra thick, dotted] (T11) -- (T12);
                \draw[ultra thick] (T11) -- (N1) (T12) -- (N2);
                \draw (T11) -- ([shift={(-1,0.5)}] T11.center) (T11) -- ([shift={(-1,-0.5)}] T11.center) (T12) -- ([shift={(1,0.5)}] T12.center) (T12) -- ([shift={(1,-0.5)}] T12.center) (N1) -- (N2);
                \node [label=below:$u$] at (T11) {};
                \node [label=below:$v$] at (T12) {};
                \node [label=above:$u^\prime$] at (N1) {};
                \node [label=above:$v^\prime$] at (N2) {};

                \node[rectangle, text width=4cm, align=center, draw=none] at (1, -1.5) {After};

                \node[token] at (T11) {};
                \node[token] at (N2) {};
            \end{scope}
        \end{tikzpicture}
        \caption{Rule~\ref{item:rule:deg4}\label{fig:rule1} applied to a link edge of degree $4$}
    \end{figure}

    \begin{figure}[!ht]
        \centering
        \begin{tikzpicture}[scale=0.8, every node/.style={circle, thick, draw, minimum size=5mm, transform shape}, token/.style={circle, thick, draw, minimum size=3mm, fill=black, transform shape}]
            \begin{scope}[shift={(0, 0)}]
                \foreach \n/\x/\y in {T11/0/0, T12/-1/1, T13/-1/-1, T21/2/0, T22/3/1, T23/3/-1}
                    {
                        \node (\n) at (\x, \y) {};
                    }
                \draw[ultra thick] (T11) -- (T12) -- (T13) -- (T11) (T21) -- (T22) -- (T23) -- (T21);
                \draw (T12) -- ([shift={(-1,0.5)}] T12.center) (T13) -- ([shift={(-1,-0.5)}] T13.center) (T22) -- ([shift={(1,0.5)}] T22.center) (T23) -- ([shift={(1,-0.5)}] T23.center);
                \draw (T11) -- (T21);
                \node [label=below:$u$] at (T11) {};
                \node [label=below:$v$] at (T21) {};

                \node[rectangle, text width=4cm, align=center, draw=none] at (1, -2) {Before};

                \node[token] at (T11) {};
                \node[token] at (T22) {};
            \end{scope}
            \begin{scope}[shift={(7, 0)}]
                \foreach \n/\x/\y in {T11/0/0, T12/-1/1, T13/-1/-1, T21/2/0, T22/3/1, T23/3/-1, N1/0.5/1, N2/1.5/1}
                    {
                        \node (\n) at (\x, \y) {};
                    }
                \draw[ultra thick] (T11) -- (T12) -- (T13) -- (T11) (T21) -- (T22) -- (T23) -- (T21) (N1) -- (N2);
                \draw (T12) -- ([shift={(-1,0.5)}] T12.center) (T13) -- ([shift={(-1,-0.5)}] T13.center) (T22) -- ([shift={(1,0.5)}] T22.center) (T23) -- ([shift={(1,-0.5)}] T23.center) (T11) -- (N1) (T21) -- (N2);
                \draw[dotted] (T11) -- (T21);
                \node [label=below:$u$] at (T11) {};
                \node [label=below:$v$] at (T21) {};
                \node [label=above:$u^\prime$] at (N1) {};
                \node [label=above:$v^\prime$] at (N2) {};

                \node[rectangle, text width=4cm, align=center, draw=none] at (1, -2) {After};

                \node[token] at (T11) {};
                \node[token] at (T22) {};
                \node[token] at (N2) {};
            \end{scope}
        \end{tikzpicture}
        \caption{Rule~\ref{item:rule:2deg3} applied to a link edge joining two degree $3$ token triangles\label{fig:rule2c1}}
    \end{figure}
    \begin{figure}[!ht]
        \centering
        \begin{tikzpicture}[scale=1.0, every node/.style={circle, thick, draw, minimum size=5mm, transform shape}, token/.style={circle, thick, draw, minimum size=3mm, fill=black, transform shape}]
            \begin{scope}[shift={(0,0)}]
                \foreach \n/\x/\y in {T11/0/0, T12/3/0, N1/1/0, N2/2/0}
                    {
                        \node (\n) at (\x, \y) {};
                    }
                \draw[ultra thick] (T11) -- (N1) (T12) -- (N2);
                \draw (T11) -- ([shift={(-1,0.5)}] T11.center) (T11) -- ([shift={(-1,-0.5)}] T11.center) (T12) -- ([shift={(1,0)}] T12.center) (N2) -- ([shift={(1,-0.7)}] N2.center) (N1) -- (N2);
                \node [label=below:$u$] at (N1) {};
                \node [label=below:$v$] at (N2) {};

                \node[rectangle, text width=4cm, align=center, draw=none] at (1, -1.5) {Before};

                \node[token] at (T11) {};
                \node[token] at (N2) {};
            \end{scope}
            \begin{scope}[shift={(6,0)}]
                \foreach \n/\x/\y in {T11/0/0, T12/3/0, N1/1/0, N2/2/0, N3/1/1, N4/2/1}
                    {
                        \node (\n) at (\x, \y) {};
                    }
                \draw[ultra thick] (T11) -- (N1) (T12) -- (N2) (N3) -- (N4);
                \draw (T11) -- ([shift={(-1,0.5)}] T11.center) (T11) -- ([shift={(-1,-0.5)}] T11.center) (T12) -- ([shift={(1,0)}] T12.center) (N2) -- ([shift={(1,-0.7)}] N2.center) (N1) -- (N3) (N2) -- (N4);
                \draw[dotted] (N1) -- (N2);
                \node [label=below:$u$] at (N1) {};
                \node [label=below:$v$] at (N2) {};
                \node [label=above:$u^\prime$] at (N3) {};
                \node [label=above:$v^\prime$] at (N4) {};

                \node[rectangle, text width=4cm, align=center, draw=none] at (1, -1.5) {After};

                \node[token] at (T11) {};
                \node[token] at (N2) {};
                \node[token] at (N3) {};
            \end{scope}
        \end{tikzpicture}
        \caption{Rule~\ref{item:rule:2deg3} applied to a link edge joining two degree $3$ token edges\label{fig:rule2c2}}
    \end{figure}
    \begin{figure}[!ht]
        \centering
        \begin{tikzpicture}[scale=1.0, every node/.style={circle, thick, draw, minimum size=5mm, transform shape}, token/.style={circle, thick, draw, minimum size=3mm, fill=black, transform shape}]
            \begin{scope}[shift={(0, 0)}]
                \foreach \n/\x/\y in {T11/0/0, T12/1/0, T21/2/0, T22/3/1, T23/3/-1}
                    {
                        \node (\n) at (\x, \y) {};
                    }
                \draw[ultra thick] (T11) -- (T12) (T21) -- (T22) -- (T23) -- (T21);
                \draw (T12) -- (T21) (T11) -- ([shift={(-1,0.5)}] T11.center) (T11) -- ([shift={(-1,-0.5)}] T11.center) (T22) -- ([shift={(1,0.5)}] T22.center) (T23) -- ([shift={(1,-0.5)}] T23.center);
                \node [label=below:$u$] at (T12) {};
                \node [label=below:$v$] at (T21) {};

                \node[rectangle, text width=4cm, align=center, draw=none] at (1, -2) {Before};

                \node[token] at (T12) {};
                \node[token] at (T22) {};
            \end{scope}
            \begin{scope}[shift={(6, 0)}]
                \foreach \n/\x/\y in {T11/0/0, T12/1/0, T21/2/0, T22/3/1, T23/3/-1, N1/1/1, N2/2/1}
                    {
                        \node (\n) at (\x, \y) {};
                    }
                \draw[ultra thick] (T11) -- (T12) (T21) -- (T22) -- (T23) -- (T21) (N1) -- (N2);
                \draw (T11) -- ([shift={(-1,0.5)}] T11.center) (T11) -- ([shift={(-1,-0.5)}] T11.center) (T22) -- ([shift={(1,0.5)}] T22.center) (T23) -- ([shift={(1,-0.5)}] T23.center) (T12) -- (N1) (N2) -- (T21);
                \draw[dotted] (T12) -- (T21);
                \node [label=below:$u$] at (T12) {};
                \node [label=below:$v$] at (T21) {};
                \node [label=above:$u^\prime$] at (N1) {};
                \node [label=above:$v^\prime$] at (N2) {};

                \node[rectangle, text width=4cm, align=center, draw=none] at (1, -2) {After};

                \node[token] at (T12) {};
                \node[token] at (T22) {};
                \node[token] at (N2) {};
            \end{scope}
        \end{tikzpicture}
        \caption{Rule~\ref{item:rule:2deg3} applied to a link edge joining a degree $3$ token edge and a degree $3$ token triangle\label{fig:rule2c3}}
    \end{figure}
    We are now ready to show that \ST remains \PSPACE-complete.
    Observe that our described construction can be done in polynomial time: each of the rules~\ref{item:rule:deg4}
    and~\ref{item:rule:2deg3} ``touches'' a token edge/link edge at most once.
    Thus, it remains to show that our construction is a valid reduction from the \ST variant used by
    Bonsma and Cereceda~\cite{BonsmaC09} to our variant.

\begin{figure}[ht]
    \centering
    \includegraphics[width=0.7\textwidth]{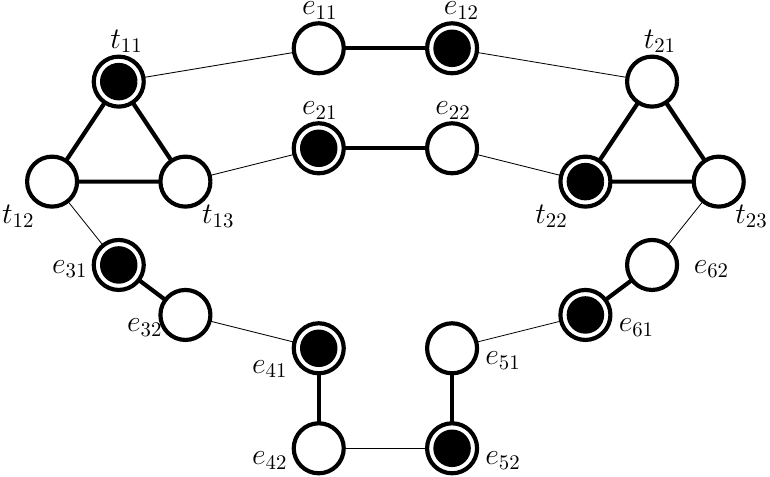}
    \caption{An example of a \ST's instance.}\label{fig:restricted-ST}
\end{figure}

    \begin{lemma}\label{lem:R1-valid}
        Let $(G, T_A, T_B)$ and $(G^1, T_A^1, T_B^1)$ be respectively the instances of \ST before and after applying~\ref{item:rule:deg4}.
        Then $(G, T_A, T_B)$ is a yes-instance if and only if $(G^1, T_A^1, T_B^1)$ is a yes-instance.
    \end{lemma}
    \begin{proof}
        Let $uv$ be some token edge of degree $4$ that is removed when
        applying~\ref{item:rule:deg4}, that is, we replace $uv$ by the path $uu^\prime v^\prime v$ where $u^\prime, v^\prime$ are newly
        added
        vertices, $uu^\prime$ and $vv^\prime$ are token edges, and $u^\prime v^\prime$ is a link edge.
        (For example, see~\cref{fig:rule1}.)
        Observe that $T_A \subset T^1_A$ and $T_B \subset T^1_B$.
        Here we use a convention that $G^1$ is constructed from $G$ by replacing the edge $uv$ by a path of length $3$.
        We note that $u \in T_A$ implies that $u, v^\prime$ are in $T^1_A$ while $u^\prime, v$ are
        not and similarly $v \in T_A$ implies that $u^\prime, v$ are in $T^1_A$ while $u, v^\prime$ are not.

        Let $\mathcal{S}$ be a $\mathsf{TS}$-sequence in $G$ between $T_A$ and $T_B$.
        We construct a $\mathsf{TS}$-sequence $\mathcal{S}^1$ in $G^1$ between $T_A^1$ and $T_B^1$ from $\mathcal{S}$ as follows.
        We replace any move $u \to v$ in $\mathcal{S}$ by the ordered sequence of moves $\langle v^\prime \to v, u \to u^\prime \rangle$
        and $v \to u$ by
        $\langle v \to v^\prime, u^\prime \to u \rangle$.
        By the construction, since the move $u \to v$ results in a new independent set of $G$, so does each member in
        $\langle v^\prime \to v, u \to u^\prime \rangle$ of $G^1$.
        More precisely, since $u \to v$ can be applied in $G$, so does $v^\prime \to v$ in $G^1$.
        After the move $v^\prime \to v$, the move $u \to u^\prime$ can be performed as no token is placed on a neighbor of $u^\prime$
        other than the one on $u$.
        Similar arguments hold for $v \to u$ and the sequence $\langle v \to v^\prime, u^\prime \to u \rangle$.
        Thus, $\mathcal{S}^1$ is indeed a $\mathsf{TS}$-sequence $\mathcal{S}^1$ in $G^1$ between $T_A^1$ and $T_B^1$.

        {
        Now, let $\mathcal{S}^1=\langle I_0, I_1, \dots, I_t\rangle$ be a $\mathsf{TS}$-sequence in $G^1$ between $T^1_A$ and $T^1_B$.
        We define a projection $\pi$ from standard token configurations of $G^1$ to standard token configurations of $G$ as follows: for every standard configuration $I$ of $G^1$, let
        \[
        \pi(I)\;:=\;\bigl(I\cap (V(G)\setminus\{u,v\})\bigr)\ \cup\ \{\,u\ \text{if}\ u\in I,\ \text{and}\ v\ \text{otherwise}\,\}.
        \]
        (Equivalently, $\pi(I)$ places the unique token of the original token edge $uv$ at $u$ if and only if the token on the token edge $uu^\prime$ is on $u$; otherwise it is placed at $v$.)

        Since $I$ is standard in $G^1$, exactly one of $u$ and $u^\prime$ is occupied.
        Moreover, if $u^\prime$ is occupied then $v^\prime$ is not (as $u^\prime v^\prime$ is a link edge), and hence the token on the token edge $vv^\prime$ must be on $v$.
        It follows that $\pi(I)$ contains exactly one of $u$ and $v$, and $\pi(I)$ is an independent set of $G$.

        Consider the projected sequence $J_0, J_1, \dots, J_t$ where $J_i:=\pi(I_i)$.
        If $J_{i-1}=J_i$, we ignore this step.
        Otherwise, $I_{i-1}$ and $I_i$ differ by a single $\mathsf{TS}$-move in $G^1$ and $J_{i-1}\neq J_i$.
        There are two cases.

        \begin{itemize}
        \item The moved token lies on a vertex of $V(G)\setminus\{u,v\}$.
        Then the same slide is an edge of $G$ (the modification only replaced the edge $uv$), and hence $J_{i-1}\to J_i$ is a valid $\mathsf{TS}$-move in $G$.

        \item The moved token lies on the token edge $uu^\prime$, i.e., the move is $u\to u^\prime$ or $u^\prime\to u$.
        If the move is $u\to u^\prime$, then $u\in I_{i-1}$ and $u^\prime\in I_i$.
        For $u\to u^\prime$ to be valid in $G^1$, the neighbor $v^\prime$ of $u^\prime$ must be unoccupied in $I_{i-1}$, and thus the token on $vv^\prime$ is on $v$ in $I_{i-1}$.
        Consequently, $J_{i-1}$ contains $u$ and $J_i$ contains $v$, so $J_{i-1}\to J_i$ is exactly the slide $u\to v$ along the original token edge $uv$ in $G$.
        Since $I_{i-1}$ is independent in $G^1$ and contains $v$, no neighbor of $v$ in $G$ is occupied, hence the slide $u\to v$ is valid in $G$.
        The case $u^\prime\to u$ is symmetric and corresponds to the slide $v\to u$ in $G$.
        \end{itemize}

        After deleting consecutive duplicates from $J_0,\dots,J_t$, we obtain a $\mathsf{TS}$-sequence in $G$ from $\pi(T^1_A)=T_A$ to $\pi(T^1_B)=T_B$.
}
    \end{proof}

    \begin{lemma}\label{lem:R2-valid}
        Let $(G^1, T_A^1, T_B^1)$ and $(G^2, T_A^2, T_B^2)$ be respectively the instances of \ST before and after applying~\ref{item:rule:2deg3}.
        Suppose that $G^1$ has no token edge of degree $4$.
        Then $(G^1, T_A^1, T_B^1)$ is a yes-instance if and only if $(G^2, T_A^2, T_B^2)$ is a yes-instance.
    \end{lemma}
    \begin{proof}
        Observe that applying~\ref{item:rule:2deg3} does not result in any new token edge of degree $4$.
        Let $uv$ be some link edge joining two degree $3$ gadgets on which~\ref{item:rule:2deg3} is applied, that is, we replace $uv$ by the path
        $uu^\prime v^\prime v$ where $u^\prime, v^\prime$ are newly added vertices, $uu^\prime$ and $v^\prime v$ are link edges, and
        $u^\prime v^\prime$ is a token edge.
        (For example, see \cref{fig:rule2c1,fig:rule2c2,fig:rule2c3}.)
        Observe that $T_A^1 \subset T_A^2$ and $T_B^1 \subset T_B^2$.
        Here we use a convention that $G^2$ is constructed from $G^1$ by replacing the edge $uv$ by a path of length $3$.
        We note that $u \in T_A^1$ implies that $u, v^\prime$ are in $T_A^2$ while $u^\prime, v$ are not and similarly $v \in T_A^1$
        implies
        that $u^\prime, v$ are in $T_A^2$ while $u, v^\prime$ are not.

        {
        Now, let $\mathcal{S}^2=\langle I_0, I_1, \dots, I_t\rangle$ be a $\mathsf{TS}$-sequence in $G^2$ between $T_A^2$ and $T_B^2$.
        Define $\pi(I):= I\cap V(G^1)$, i.e., we forget the two new vertices $u^\prime$ and $v^\prime$.

        We first show that no move of $\mathcal{S}^2$ uses the link edges $uu^\prime$ or $vv^\prime$, and hence the only moves involving $\{u^\prime,v^\prime\}$ are slides along the token edge $u^\prime v^\prime$.
        Indeed, throughout $\mathcal{S}^2$ every token triangle and every token edge contains exactly one token (as $\mathcal{S}^2$ starts from a standard token configuration and no valid token-slide can traverse a link edge into a gadget that already contains a token adjacent to the entry vertex).
        In particular, the gadget containing $u$ contains a token at all times. If $u\notin I_{i}$, then the unique token of $u$'s gadget occupies a neighbor of $u$ (since token triangles are cliques and token edges are edges), and therefore $u$ cannot be the target of a valid slide at step $i$. Thus, no move $u^\prime\to u$ can occur.
        If $u\in I_i$, then $u^\prime\notin I_i$ (since $uu^\prime$ is a link edge), so the unique token on the token edge $u^\prime v^\prime$ must be on $v^\prime$. Hence $u^\prime$ is adjacent to an occupied vertex ($v^\prime$), and therefore no move $u\to u^\prime$ can occur.
        The same argument applies symmetrically to the link edge $vv^\prime$. 
        Consequently, the only possible move that touches $u^\prime$ or $v^\prime$ is the slide along the token edge $u^\prime v^\prime$.

        Moreover, for every configuration $I$ reachable from $T_A^2$, the vertices $u$ and $v$ are never simultaneously occupied:
        if $u\in I$ then $u^\prime\notin I$, so the token on $u^\prime v^\prime$ must be on $v^\prime$, which is adjacent to $v$, implying $v\notin I$ (and symmetrically).
        Consequently, $\pi(I)$ is an independent set of $G^1$ (in particular, it respects the original link edge $uv$).

        Consider the projected sequence $J_0,\dots,J_t$ where $J_i:=\pi(I_i)$.
        If $I_{i-1}\to I_i$ is the move $u^\prime\leftrightarrow v^\prime$, then $J_{i-1}=J_i$.
        Otherwise, the moved token lies on $V(G^1)$ and the same slide is an edge of $G^1$ (all gadgets of $G^1$ are preserved in $G^2$), hence $J_{i-1}\to J_i$ is a valid $\mathsf{TS}$-move in $G^1$.
        After deleting consecutive duplicates from $J_0,\dots,J_t$, we obtain a $\mathsf{TS}$-sequence in $G^1$ from $\pi(T_A^2)=T_A^1$ to $\pi(T_B^2)=T_B^1$.
        }
    \end{proof}

    Combining \cref{lem:R1-valid,lem:R2-valid}, we have the following.
    \begin{lemma}\label{lem:R1R2-valid}
    Let $(G,T_A,T_B)$ be an instance of \ST. Let $(G^\prime,T_A^\prime,T_B^\prime)$ be the instance obtained by exhaustively applying Rules~\ref{item:rule:deg4} and~\ref{item:rule:2deg3} as described above (adding tokens so that each intermediate instance remains standard).
    Then $(G,T_A,T_B)$ is a yes-instance if and only if $(G^\prime,T_A^\prime,T_B^\prime)$ is a yes-instance.
    \end{lemma}
    \begin{proof}
    Consider the sequence of intermediate instances produced by the exhaustive procedure:
    \[
    (G_0,T_A^{(0)},T_B^{(0)}),\ (G_1,T_A^{(1)},T_B^{(1)}),\ \dots,\ (G_m,T_A^{(m)},T_B^{(m)}),
    \]
    where $(G_0,T_A^{(0)},T_B^{(0)})=(G,T_A,T_B)$ and $(G_m,T_A^{(m)},T_B^{(m)})=(G^\prime,T_A^\prime,T_B^\prime)$, and each transition applies exactly one rule once.
    If $G_{i+1}$ is obtained from $G_i$ by applying Rule~\ref{item:rule:deg4}, then yes-instances are preserved by \cref{lem:R1-valid}.
    If $G_{i+1}$ is obtained from $G_i$ by applying Rule~\ref{item:rule:2deg3}, then (since Rule~\ref{item:rule:deg4} has already been exhausted and no degree-$4$ token edge is created) yes-instances are preserved by \cref{lem:R2-valid}.
    Chaining these equivalences over $i=0,\dots,m-1$ yields the claim.
    \end{proof}

    Combining our construction and \cref{lem:R1R2-valid}, we are now ready to prove \cref{thm:rst-pspacec}.

\thmrstpspacec*

\begin{proof}
    Let $(G, T_A, T_B)$ be an instance of \textsc{\ST} and $(G^\prime, T^\prime_A, T^\prime_B)$ be the corresponding instance of \ST.
    Our construction and \cref{lem:R1R2-valid} imply the \PSPACE-completeness.
    From our construction, since the input graph $G$ is planar, so is the constructed graph $G^\prime$.
    
    {
    Additionally, we show that $G^\prime$ is $2$-degenerate.
    Recall that (by construction) $G^\prime$ has maximum degree~$3$.
    It therefore suffices to show that no induced subgraph of $G^\prime$ can have minimum degree at least~$3$.

    Suppose for contradiction that $X$ is a nonempty induced subgraph of $G^\prime$ with $\delta(X)\ge 3$.
    Since $\Delta(G^\prime)\le 3$, every vertex of $X$ has degree exactly~$3$ in $G^\prime$, and in particular every neighbor of any vertex of $X$ in $G^\prime$ also belongs to $X$.

    Let $x\in V(X)$.
    If $x$ is a degree-$3$ vertex in a token edge gadget, then by construction all neighbors of $x$ in $G^\prime$ have degree at most~$2$ in $G^\prime$.
    As these neighbors also lie in $X$, at least one vertex of $X$ has degree at most~$2$ in $X$, contradicting $\delta(X)\ge 3$.

    Otherwise, $x$ is a vertex of a token triangle.
    In that case, $x$ has two neighbors inside the triangle and one neighbor $y$ outside the triangle.
    By construction, this outside neighbor $y$ has degree at most~$2$ in $G^\prime$.
    As argued above, $y\in V(X)$, and hence $y$ has degree at most~$2$ in $X$, again contradicting $\delta(X)\ge 3$.

    This contradiction shows that every induced subgraph of $G^\prime$ has a vertex of degree at most~$2$, and thus $G^\prime$ is $2$-degenerate.}

\end{proof}

\subsection{Reduction to List $(d, k)$-\CR}\label{sec:list-dk-col-reconf}
In this section, we describe a reduction from our variant of \ST to \textsc{List $(d, k)$-\CR}. In particular we show,

\begin{restatable}{theorem}{thmlistdkpspacec}\label{thm:listdk-pspacec}
    \textsc{List $(d, k)$-\CR} is \PSPACE-complete even on planar, bipartite and $2$-degenerate graphs, for any fixed $d \geq 2$ and $k \geq {3d+3}$ if $d$ is odd or $k\geq {3d+6}$
    if $d$ is even.
\end{restatable}

Recall that in a list $(d, k)$-coloring of a graph $G$, each vertex $v$ is associated with a list $L(v) \subseteq \{1, \dots, k\}$ of colors that it can have, and no two vertices whose distance is at most $d$ in $G$ share the same color.

\subsubsection{Forbidding Paths}
We begin by defining an analogous concept of the ``$(a, b)$-forbidding path'' defined in~\cite{BonsmaC09}.
Intuitively, in such paths, their endpoints can never at any step be respectively colored $a$ and $b$.
This is useful for simulating the behavior of token movements: both endpoints of an edge cannot have tokens simultaneously.
We augment the original definition with a set of colors $C$.

\begin{definition}\label{def:ab-forbid}
    Let $u, v$ be two vertices of a graph $G$, $d \geq 2$ and $k \geq d+5$ be fixed integers.
    Let $a, b \in \{1, 2, 3\}$ and $C \subseteq \{1,\dots,k\}\setminus \{1,2,3\}$ be a given set of colors. {(In our reduction, $|C|$ will be chosen large enough to accommodate additional buffer colors used to separate base colors near high-degree vertices.)}
    For a $uv$-path $P$ and a $(d, k)$-coloring $\alpha$ of $P$, we call $\alpha$ an $[x, y]$-coloring if $\alpha(u) = x$ and $\alpha(v) = y$.
    A \textit{$(C, a, b)$-forbidding path} from $u$ to $v$ is a $uv$-path $P$ in $G$ with color list $L$ such that both $L(u)$ and
    $L(v)$ are subsets of
    $\{1, 2, 3\}$, $a \in L(u)$, $b \in L(v)$, $\bigcup_{w \in V(P) \setminus \{u, v\}}L(w) \subseteq (C \cup \{a, b\})$, and the following two conditions are satisfied:
    \begin{enumerate}[label=(\arabic*)]
        \item An $[x, y]$-coloring exists if and only if $x \in L(u)$, $y \in L(v)$, and $(x, y) \neq (a, b)$.
              Such a pair $(x, y)$ is called \textit{admissible} for $P$.\label{def:ab-forbid:item1}

        \item If both $(x, y)$ and $(x^\prime, y)$ are admissible, then from any $[x, y]$-coloring, there exists a reconfiguration
              sequence that ends with a $[x^\prime, y]$-coloring, without ever recoloring $v$, and only recoloring $u$ in the last step.
              A similar statement holds for admissible pairs $(x, y)$ and $(x, y^\prime)$.\label{def:ab-forbid:item2}
    \end{enumerate}
\end{definition}
As in~\cite{BonsmaC09}, a $(C, a, b)$-forbidding path $P$ between vertices $u$ and $v$ serves to prevent them from simultaneously having colors $a$ and $b$, respectively. Any other combination of colors for $u$ and $v$ from their respective color lists remains valid. Furthermore, the vertices can be recolored to any permissible colors as long as they avoid the forbidden combination. It's worth noting that a $(C, a, b)$-forbidding path from $u$ to $v$ is distinct from a $(C, a, b)$-forbidding path from $v$ to $u$.

\begin{figure}[ht]
    \centering
    \includegraphics[width=0.8\textwidth]{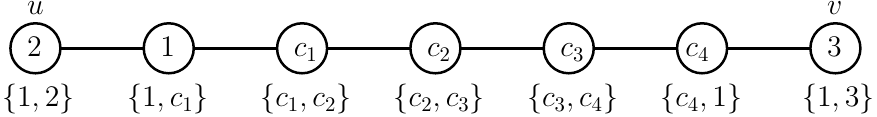}
    \caption{An example of a $(C, 1, 1)$-forbidding path $P$ between two vertices $u$ and $v$ having $L(u) = \{1, 2\}$ and $L(v) = \{1, 3\}$. Here $d = 2$, $k = 7$ ($= d + 5$), and $C = \{c_1, c_2, c_3, c_4\}$. The numbers inside the vertices of $P$ indicate a $(2, 7)$-coloring $\alpha$ of $P$ which is also a $[2, 3]$-coloring.}\label{fig:ab-forbidding-paths}
\end{figure}

In the following lemma, we demonstrate the construction of a $(C, a, b)$-forbidding path for any $d \geq 2$. The path has length $\ell$, where $\ell = d + 3$ if $d$ is odd and $\ell = d + 4$ if $d$ is even. These specific length formulations ensure that our constructed path $P$ always has even length, which is crucial for having the bipartite property in our subsequent graph construction.

\begin{lemma}\label{lem:forbidpath}
    Let $d \geq 2$ and $k \geq d+5$.
    Let $C$ be a given set of colors such that $C \cap \{1, 2, 3\} = \emptyset$ and $|C|$ is either $ d + 1$ if $d$ is odd or $d+2$ if $d$ is
    even.
    For any $L_u \subseteq \{1, 2, 3\}$, $L_v \subseteq \{1, 2, 3\}$, $a \in L_u$, and
    $b \in L_v$, there exists a $(C, a, b)$-forbidding path $P$ with $L(u) = L_u$, $L(v) = L_v$ and for any
    $w \in V(P) \setminus \{u, v\}$, $L(w) \subseteq (C \cup \{a, b\})$.
    Moreover, $P$ has length $d+3$ if $d$ is odd and $d+4$ if $d$ is even.
\end{lemma}
\begin{proof}
    Suppose that $C = \{c_1, \dots, c_p\}$ where $p$ is either $d + 1$ if $d$ is odd or $d+2$ if $d$ is even.
    We define the path $P = v_0v_1\dots v_p v_{p+1} v_{p+2}$ such that $v_0 = u$ and $v_{p+2} = v$.
    $P$ has length $p + 2$, which is equal to $(d + 1) + 2 = d + 3$ if $d$ is odd and $(d + 2) + 2 = d + 4$ if $d$ is even.
    We define the color list $L$ for each vertex of $P$ as follows.
    \begin{itemize}
        \item $L(u) = L(v_0) = L_u$, $L(v) = L(v_{p+2}) = L_v$, $L(v_1) = \{a, c_1\}$, and $L(v_{p+1}) = \{c_p, b\}$.
        \item For $2 \leq i \leq p$, $L(v_i) = \{c_{i-1}, c_i\}$.
    \end{itemize}
    We show that the path $P$ with the color list $L$ indeed form a $(C, a, b)$-forbidding path.
    It suffices to verify the conditions~\ref{def:ab-forbid:item1} and~\ref{def:ab-forbid:item2} in \cref{def:ab-forbid}.

    We first verify~\ref{def:ab-forbid:item1}.
    Suppose that $x \in L(u)$, $y \in L(v)$, and $(x, y) \neq (a, b)$.
    We describe how to construct a $[x, y]$-coloring.
    If $x = a$, we color $v_0$ by $a$, $v_{p+1}$ by $b$, $v_{p+2}$ by $y$ and $v_i$ by $c_i$ for $1 \leq i \leq p$.
    Similarly, if $y = b$, we color $v_{p+2}$ by $b$, $v_1$ by $a$, $v_0$ by $x$, and $v_i$ by $c_{i-1}$ for $2 \leq i \leq p+1$.
    If both $x \neq a$ and $y \neq b$, one possible valid coloring is to color $v_0$ by $x$, $v_1$ by $a$, $v_{p+1}$ by $b$, $v_{p+2}$ by
    $y$ and $v_i$ by $c_i$ for $2 \leq i \leq p$.

    {
    We now verify that each explicit assignment above is a list $(d,k)$-coloring of $P$.
    First, every color $c_i\in C$ is used on exactly one vertex of $P$ in each of the three constructions, so no conflict can involve any color of $C$.
    Thus, any potential conflict must involve a base color in $\{1,2,3\}$.

    In each construction, the only vertices that may carry a base color are the endpoints $v_0=u$ and $v_{p+2}=v$ and the two vertices $v_1$ and $v_{p+1}$.
    If a base color appears twice, then either it appears on the endpoints (whose distance is $p+2\ge d+3>d$), or it appears on $v_1$ and $v_{p+1}$ (which can only happen when $a=b$, in which case their distance is $p\ge d+1>d$), or it appears on an endpoint and one of $\{v_1,v_{p+1}\}$ (in which case the distance is at least $p+1\ge d+2>d$).
    Hence no two vertices at distance at most $d$ receive the same color, and the constructed assignment is a valid $(d,k)$-coloring of $P$.}

    On the other hand, suppose that a $[x, y]$-coloring of $P$ exists.
    We claim that $x \in L(u)$, $y \in L(v)$, and $(x, y) \neq (a, b)$.
    The first two conditions are followed from the definition of a $[x, y]$-coloring.
    We show that the last condition holds.
    Observe that if $u = v_0$ has color $x = a$ then we are forced to color $v_1$ by $c_1$,
    $v_2$ by $c_2$, and so on until  $v_p$ by $c_p$, and $v_{p+1}$ by $b$, which implies that the color
    of $v_{p+2} = v$ cannot be $b$; {otherwise $v_{p+1}$ is forced to be colored $b$ while $v_{p+2}=v$ is also colored $b$, and since $\dist_P(v_{p+1},v_{p+2})=1\le d$, this violates the distance-$d$ constraint.}
    Similar arguments can be applied for the case $v = v_{p+2}$ has color $y = b$.
    Thus, $(x, y) \neq (a, b)$.

    We now verify~\ref{def:ab-forbid:item2}.
    Let $(x, y)$ and $(x, y^\prime)$ be two admissible pairs.
    From~\ref{def:ab-forbid:item1}, a $[x, y]$-coloring $\alpha$ and a $[x, y^\prime]$-coloring $\beta$ of $P$ exist.
    We describe how to construct a reconfiguration sequence $\mathcal{S}$ which starts from $\alpha$, ends at $\beta$, and
    satisfies~\ref{def:ab-forbid:item2}.
    If $v_0$ has color $x = a$, then both $\alpha$ and $\beta$ have the same coloring for all vertices of $P$ except at vertex $v$ and are
    therefore adjacent $(d, k)$-colorings.
    As $(x, y^\prime)$ is an admissible pair, $y' \neq b$, hence, we can recolor $v$ from $y$ to $y'$.

    If $v_0$ has color $x \neq a$, we show that a reconfiguration sequence from $\alpha$ to
    $\beta$ exists by describing a procedure that recolors both $\alpha$ and $\beta$ to the same $[x, y^\prime]$-coloring $\gamma$ of $P$.
    The coloring $\gamma$ is constructed from any $[x, y]$-coloring of $P$ where $x \neq a$ as follows.
    First, we recolor $v_1$ by $a$ (if it was already colored $a$ then there is nothing to do).
    Notice that the only other vertex in $v_1,\dots,v_{p+1}$ which can potentially have the color $a$ is $v_{p+1}$, in the case that
    $a=b$. But as the distance between $v_1$ and $v_{p+1}$ is $p>d$, this recoloring step is valid.

    Next, we recolor $v_2$ by $c_1$ as $c_1$ is not used to color any other vertex in $P$ currently.
    We proceed with coloring every $v_i$ by color $c_{i-1}$, $3 \leq i \leq p+1$.
    Each recoloring step above is valid, as when we color $v_i$ by $c_{i-1}$, we always ensure that $c_{i-1}$ is not used to color any
    other vertex in $P$ at that time.
    Again, if $v_i$ is already colored $c_{i-1}$ then there is nothing to do.
    At the end of this process $v_{p+1}$ is colored with $c_p$.
    Finally, as the nearest vertex to $v = v_{p+2}$ which is colored $a$ is the vertex $v_1$ is at distance $p + 1 > d$ from $v$, this
    leaves us free to color $v$ with $y' \in L(v)$.
    This gives the required reconfiguration sequence from $\alpha$ to $\beta$ by combining the sequences from $\alpha$ to $\gamma$ and
    from $\beta$ to $\gamma$.
    The case for admissible pairs $(x,y)$ and $(x',y)$ is symmetric.
\end{proof}

\begin{lemma}\label{lem:forbidpath-padding}
    Let $d \ge 2$ and let $P_0$ be a $(C_0,a,b)$-forbidding path from $u$ to $v$ (in the sense of \cref{def:ab-forbid}).
    Fix an integer $r$ with $1 \le r \le d-1$, and let $D_u=\{d^u_1,\dots,d^u_r\}$ and $D_v=\{d^v_1,\dots,d^v_r\}$ be two disjoint color sets such that
    $(D_u \cup D_v)\cap(\{1,2,3\}\cup C_0)=\emptyset$.
    Construct a new $uv$-path $P$ from $P_0$ by subdividing the edge incident to $u$ by $r$ new vertices $s_1,\dots,s_r$ (in this order from $u$),
    and subdividing the edge incident to $v$ by $r$ new vertices $t_1,\dots,t_r$ (in this order towards $v$).
    Extend the list assignment of $P_0$ to $P$ by setting $L(s_i)=\{d^u_i\}$ for all $i\in\{1,\dots,r\}$ and $L(t_i)=\{d^v_i\}$ for all $i\in\{1,\dots,r\}$.
    Then $P$ is a $(C,a,b)$-forbidding path from $u$ to $v$ for $C:=C_0\cup D_u\cup D_v$.
\end{lemma}
\begin{proof}
    Since each new vertex has a singleton list, its color is fixed in every list $(d,k)$-coloring of $P$.
    Moreover, all newly introduced colors are distinct and do not belong to $\{1,2,3\}$.

    We verify Conditions~\ref{def:ab-forbid:item1} and~\ref{def:ab-forbid:item2}.
    First, observe that subdividing edges only increases graph distances among vertices of $P_0$.
    Hence, any list $(d,k)$-coloring of $P_0$ remains a valid list $(d,k)$-coloring when viewed on the corresponding vertices of $P$.
    Conversely, restricting any list $(d,k)$-coloring of $P$ to the vertices of $P_0$ yields a list $(d,k)$-coloring of $P_0$.

    For Condition~\ref{def:ab-forbid:item1}, fix $x\in L(u)$ and $y\in L(v)$.
    By the above correspondence, an $[x,y]$-coloring exists for $P$ if and only if an $[x,y]$-coloring exists for $P_0$.
    Since $P_0$ is a $(C_0,a,b)$-forbidding path, this holds if and only if $(x,y)\neq(a,b)$.

    For Condition~\ref{def:ab-forbid:item2}, let $(x,y)$ and $(x',y)$ be admissible pairs.
    Take any $[x,y]$-coloring $\alpha$ of $P$ and restrict it to the vertices of $P_0$.
    Applying Condition~\ref{def:ab-forbid:item2} in $P_0$ gives a reconfiguration sequence on $P_0$ that keeps $v$ fixed to $y$ and recolors $u$ only in the last step, ending in a coloring from which $u$ can be recolored to $x'$.
    We lift this sequence to $P$ by performing the same recolorings on the corresponding vertices of $P_0$ and leaving all new vertices fixed.
    Every intermediate coloring remains valid in $P$ because distances between vertices of $P_0$ weakly increase and the new fixed colors do not occur on any vertex of $P_0$.
    The symmetric statement for changing $v$ follows analogously.
\end{proof}

\subsubsection{Construction of Our Reduction}
We are now ready to describe our reduction.
Let $(G, T_A, T_B)$ be an instance of \ST.
We describe how to construct a corresponding instance $(G^\prime, \alpha, \beta, L)$ of \textsc{List $(d, k)$-\CR}.
We use the same notations in~\cite{BonsmaC09} to label the vertices of $G$.
The token triangles are labelled $1, \dots, n_t$, and the vertices of the triangle $i$ are $t_{i1}$, $t_{i2}$, and $t_{i3}$.
The token edges are labelled $1, \dots, n_e$, and the vertices of the token edge $i$ are $e_{i1}$ and $e_{i2}$.
To construct $G^\prime$ and $L$, we proceed as follows.

For every token triangle $i$ ($1 \leq i \leq n_t$), we introduce a vertex $t_i$ in $G^\prime$ with color list $L(t_i) = \{1, 2, 3\}$.
For every token edge $j$ ($1 \leq j \leq n_e$), we introduce a vertex $e_j$ in $G^\prime$ with color list $L(e_j) = \{1, 2\}$.
From our construction of \ST, in $G^\prime$, each $t_i$ has degree exactly three and each $e_j$ has degree either two
or three.
Whenever a link edge of $G$ joins a vertex $t_{ia}$ ($1 \leq i \leq n_t$) with a vertex $e_{jb}$ ($1 \leq j \leq n_e$) or it
joins
$e_{ia}$ ($1 \leq i \leq n_e$) with $e_{jb}$ ($1 \leq j \leq n_e$), we define $u = t_i$ and $v = e_j$ if we consider $\{t_{ia}, e_{jb}\}$, and $u = e_i$ and $v = e_j$ if we consider $\{e_{ia}, e_{jb}\}$,
and add to $G^\prime$ a $(C_{uv}, a, b)$-forbidding path
$P_{uv} = w_{uv}^0w_{uv}^1\dots w_{uv}^p$ of length $p$
between $u= w_{u v}^0$ and $v = w_{uv}^p$ in $G^\prime$,
where {$r := \lfloor d/2\rfloor$ and $p := d + 3 + 2r$ if $d$ is odd and $p := d + 4 + 2r$ if $d$ is even. (Equivalently, $p = 2d+2$ if $d$ is odd and $p = 2d+4$ if $d$ is even.)}
$C_{uv}$ is the set of exactly $p - 2$ colors which we will define later along with the color list $L$ for each vertex in $P_{uv}$.
(We remark that, unlike in~\cite{BonsmaC09}, our construction of \ST guarantees that there is no link
edge joining a $t_{ia}$ ($1 \leq i \leq n_t$) with a $t_{jb}$ ($1 \leq j \leq n_t$).)

Let $q= (p-2)/2$. By definition, $p \geq d+3 \geq 4$ and $p$ is always even, which means $q \geq 1$ and $q \in \mathbb{N}$.
For each forbidding path $P_{uv} = w_{uv}^0\dots w_{uv}^p$, we partition its vertex set into two \textit{parts}: the \textit{closer part} (from $u$ than $v$) denoted by
$\cl(P_{uv}) = \{w_{uv}^0, \dots, w_{uv}^{q+1}\}$ and
the \textit{further part} (from $u$ than $v$) denoted by $\far(P_{uv}) = \{w_{uv}^{q+1}, \dots, w_{uv}^{p}\}$.
Note that for a forbidding path $P_{uv}$, the two
\textit{parts} $\cl(P_{uv})$ and $\far(P_{uv})$ intersect at exactly
one vertex, namely $w_{uv}^{q+1}$.
Additionally, $\cl(P_{uv}) = \far(P_{vu})$ and $\far(P_{uv}) = \cl(P_{vu})$.
We say that a \textit{part} $\cl(P_{uv})$ which contains $u = w_{u v}^0$ is \textit{incident} to $u$ and similarly $\far(P_{uv})$ which
contains $v = w_{uv}^p$ is \textit{incident} to $v$.
From our construction of \ST, each $t_i$ has exactly three \textit{parts} incident to it and each $e_j$ has either two
or three \textit{parts} incident to it.
(Recall that $u, v \in \{t_i, e_j\}$.)

To construct the set $C_{uv}$ and the list $L$ for each vertex of $P_{uv}$, we will use three disjoint sets $A, B, C$ of colors.
Each set $A, B$ or $C$ is an ordered set of colors of size $q$ and has no common member with $\{1, 2, 3\}$.
For $\part \in \{\cl, \far\}$, let $f\colon \part(P_{uv}) \to \{A, B, C\}$ be a function which assigns exactly one set of colors in
$\{A, B, C\}$ to each \textit{part} of these paths $P_{uv}$ such that:
\begin{enumerate}[label=(P\arabic*)]
    \item No two \textit{parts} of the same forbidding path share the same assigned set, i.e., $f(\cl(P_{uv})) \neq f(\far(P_{uv}))$; and\label{item:forbidden-color:a}
    \item No two \textit{parts} incident to the same vertex in $G^\prime$ share the same
          assigned set, i.e., for any pair $v, v^\prime$ of $u$'s neighbors,
          $f(\cl(P_{uv})) \neq f(\cl(P_{u v^\prime}))$.\label{item:forbidden-color:b}
\end{enumerate}

In the rest of the proof, we refer to the conditions above as
conditions~\ref{item:forbidden-color:a} and~\ref{item:forbidden-color:b}
respectively. We will show later in \cref{lem:fmapping-exists} that such a function can be constructed in polynomial time.
After we use the function $f$ to assign the colors $\{A,B, C\}$ to \textit{parts} of a forbidding path $P_{u v}$, we are ready to define
$C_{uv}$.
Suppose that the ordered set $X = (x_1, \dots, x_q) \in \{A,B,C\}$ is used to color $\cl(P_{uv})$ and the ordered set $Y = (y_1, \dots, y_q) \in \{A,B,C\} \setminus X$ is used to color
$\far(P_{uv})$, that is, $X = f(\cl(P_{uv}))$ and $Y = f(\far(P_{uv}))$.
We define $C_{uv}= X \cup Y$.
Next, we define the color list $L$ for a path $P_{uv} = w_{uv}^0\dots w_{uv}^p$ (where $u = w_{uv}^0$ and $v = w_{uv}^p$) using colors $C_{uv}$, as follows.

\begin{itemize}
    \item If $u = t_i$ for some $i \in \{1, \dots, n_t\}$, define $L(u) = \{1, 2, 3\}$; otherwise (i.e., $u = e_j$ for some $j \in \{1, \dots, n_e\}$), define $L(u) = \{1, 2\}$.
          Similar definitions hold for $L(v)$.
    \item {Let $r:=\lfloor d/2\rfloor$. For each $1\le i\le r$, set $L(w_{uv}^i)=\{x_i\}$ and $L(w_{uv}^{p-i})=\{y_i\}$. (These are buffer vertices whose colors are fixed.)}
    \item {$L(w_{uv}^{r+1})=\{a,x_{r+1}\}$ and $L(w_{uv}^{p-(r+1)})=\{y_{r+1},b\}$.}
    \item {For $r+2\le i\le q$, set $L(w_{uv}^i)=\{x_{i-1},x_i\}$ and $L(w_{uv}^{p-i})=\{y_i,y_{i-1}\}$, and set $L(w_{uv}^{q+1})=\{x_q,y_q\}$.}
\end{itemize}

{By \cref{lem:forbidpath-padding} (applied to the core construction of \cref{lem:forbidpath}), each $P_{uv}$ is a $(C_{uv},a,b)$-forbidding path, and the base colors $\{1,2,3\}$ can only appear on vertices at distance at least $r+1$ from either endpoint.}

{While the buffer vertices on forbidding paths are assigned fixed colors, these colors are chosen from the palettes $A,B,C$ in a way that prevents any \emph{distance-$d$ interference} across different paths.
Indeed, if two forbidding paths share an endpoint, then Condition~(P2) forces their endpoint-incident parts to use distinct palettes, and hence their fixed colors are automatically different.
If two forbidding paths do not share an endpoint, then any walk between internal vertices of the two paths must traverse at least one complete forbidding path of length at least $2d+2$, and therefore has length greater than $d$.
The next lemma makes this separation precise and will allow us to treat recolorings performed inside a single forbidding path as \emph{local} operations in the ambient graph $G^\prime$.}

\begin{lemma}\label{lem:palette-separation}
Let $\chi\in A\cup B\cup C$ and let $x,y\in V(G^\prime)$ be two distinct vertices such that $\chi\in L(x)\cap L(y)$.
If $x$ and $y$ do not lie on the same forbidding path $P_{uv}$, then $\dist_{G^\prime}(x,y)>d$.
\end{lemma}
\begin{proof}
Since $\chi\in A\cup B\cup C$, both $x$ and $y$ are internal vertices of forbidding paths (endpoints $t_i$ and $e_j$ have lists contained in $\{1,2,3\}$).
Let $x$ lie on a forbidding path $P_{uv}$ and let $y$ lie on a forbidding path $P_{u'v'}$ with $P_{uv}\neq P_{u'v'}$.

{Assume first that the two paths share an endpoint, and let this common endpoint be $u$.
Write the two paths as $P_{uv}$ and $P_{uv'}$ with $P_{uv}\neq P_{uv'}$.
Let $X=f(\cl(P_{uv}))$ and $Y=f(\far(P_{uv}))$, and let $X'=f(\cl(P_{uv'}))$ and $Y'=f(\far(P_{uv'}))$.
By Condition~\ref{item:forbidden-color:b}, we have $X\neq X'$.
Recall that $q=(p-2)/2\ge d$, and by the definition of the lists on forbidding paths, every vertex of $P_{uv}$ whose list contains a color from $Y$ is at distance at least $q+1\ge d+1$ from $u$ (and the same holds for $P_{uv'}$ with $Y'$).

Now suppose that $\chi\in L(x)\cap L(y)$.
If $\chi\in X$, then since $X'\neq X$, any vertex of $P_{uv'}$ whose list contains $\chi$ lies in $\far(P_{uv'})$ and hence is at distance at least $q+1\ge d+1$ from $u$.
Similarly, if $\chi\in X'$, then any vertex of $P_{uv}$ whose list contains $\chi$ lies in $\far(P_{uv})$.
If $\chi\notin X\cup X'$, then both $x$ and $y$ lie in their respective further parts.
In all cases, at least one of $x$ and $y$ is at distance at least $d+1$ from $u$, while the other is an internal vertex and hence at distance at least $1$ from $u$.

Any $xy$-path in $G^\prime$ must leave $P_{uv}$ through $u$ or $v$ and enter $P_{uv'}$ through $u$ or $v'$.
If it uses $u$, then its length is at least $\dist_{G^\prime}(x,u)+\dist_{G^\prime}(u,y)\ge 1+(d+1)>d$.
Otherwise, it contains a subpath between two distinct skeleton vertices and therefore has length at least $p\ge 2d+2>d$.
Thus, in the shared-endpoint case we have $\dist_{G^\prime}(x,y)>d$.}

We may therefore assume that $P_{uv}$ and $P_{u'v'}$ do not share endpoints.
Every internal vertex of a forbidding path has degree $2$ in $G^\prime$.
Thus, any $xy$ path in $G^\prime$ must leave $P_{uv}$ through one of its endpoints and enter $P_{u'v'}$ through one of its endpoints.
Since every adjacency between skeleton vertices in $G^\prime$ is realized by a forbidding path of length $p\ge 2d+2$, we obtain
\[
\dist_{G^\prime}(x,y)\ \ge\ 1+p+1\ \ge\ 2d+4\ >\ d,
\]
as claimed.
\end{proof}

{\cref{lem:palette-separation} implies that whenever a vertex of a forbidding path is recolored to an auxiliary color in $A\cup B\cup C$, no vertex outside the path can lie within distance at most $d$ and simultaneously carry the same color.
Consequently, a reconfiguration sequence designed for a forbidding path remains valid when executed inside $G^\prime$ while keeping all other vertices fixed.}

\begin{corollary}\label{cor:path-locality}
Let $P_{uv}$ be a forbidding path of $G^\prime$.
Let $\mathcal{R}=\langle \varphi_0,\varphi_1,\dots,\varphi_m\rangle$ be a reconfiguration sequence on $P_{uv}$ such that
\begin{enumerate}[label=(\roman*)]
    \item $\varphi_i(u)=\varphi_0(u)$ and $\varphi_i(v)=\varphi_0(v)$ for all $i$, and
    \item every step recolors a vertex of $P_{uv}\setminus\{u,v\}$.
\end{enumerate}
Then $\mathcal{R}$ is also a valid reconfiguration sequence in $G^\prime$ when all vertices of $V(G^\prime)\setminus V(P_{uv})$ are kept fixed.
\end{corollary}
\begin{proof}
Consider a recoloring step that changes a vertex $x\in V(P_{uv})\setminus\{u,v\}$ to a color $\gamma\in L(x)$.
We must show that there is no vertex $y\notin V(P_{uv})$ with $\varphi(y)=\gamma$ and $\dist_{G^\prime}(x,y)\le d$ (where $\varphi$ denotes the current coloring outside the step).

If $\gamma\in A\cup B\cup C$, then any vertex $y\notin V(P_{uv})$ with color $\gamma$ satisfies $\gamma\in L(y)$, and \cref{lem:palette-separation} gives $\dist_{G^\prime}(x,y)>d$.

Assume now that $\gamma\in\{1,2,3\}$.
By construction of forbidding paths (cf.\ the discussion in the proof of \cref{lem:list-dkcol-const}), the only vertices of $G^\prime$ whose lists contain a base color are the skeleton vertices and, for each forbidding path $P_{ab}$, the two core vertices at distance exactly $r+1$ from $a$ and from $b$.
Since $x$ is an internal vertex of $P_{uv}$ and is being recolored to a base color, $x$ must be one of these two core vertices of $P_{uv}$; in particular,
\[
\min\{\dist_{G^\prime}(x,u),\dist_{G^\prime}(x,v)\}=r+1.
\]
Let $y\notin V(P_{uv})$ be any vertex with color $\gamma$.
Then $y$ is either a skeleton vertex distinct from $u,v$, or a core vertex of some other forbidding path $P_{ab}\neq P_{uv}$.
If $P_{ab}$ shares an endpoint with $P_{uv}$, say $a=u$, then $\dist_{G^\prime}(u,y)=r+1$ by the same reasoning, and therefore
\[
\dist_{G^\prime}(x,y)\ \ge\ \dist_{G^\prime}(x,u)+\dist_{G^\prime}(u,y)\ \ge\ (r+1)+(r+1)\ =\ 2(r+1)\ >\ d
\]
(using $r=\lfloor d/2\rfloor$).
Otherwise, any $x$--$y$ path must traverse at least one complete forbidding path between two distinct skeleton vertices, and hence has length at least $1+p+1>d$ because $p\ge 2d+2$.
Thus, $\dist_{G^\prime}(x,y)>d$ in all cases.

Therefore the recoloring step remains valid in $G^\prime$, and the claim follows by induction over the steps of $\mathcal{R}$.
\end{proof}

\begin{figure}[!ht]
    \centering
    \includegraphics[width=\textwidth]{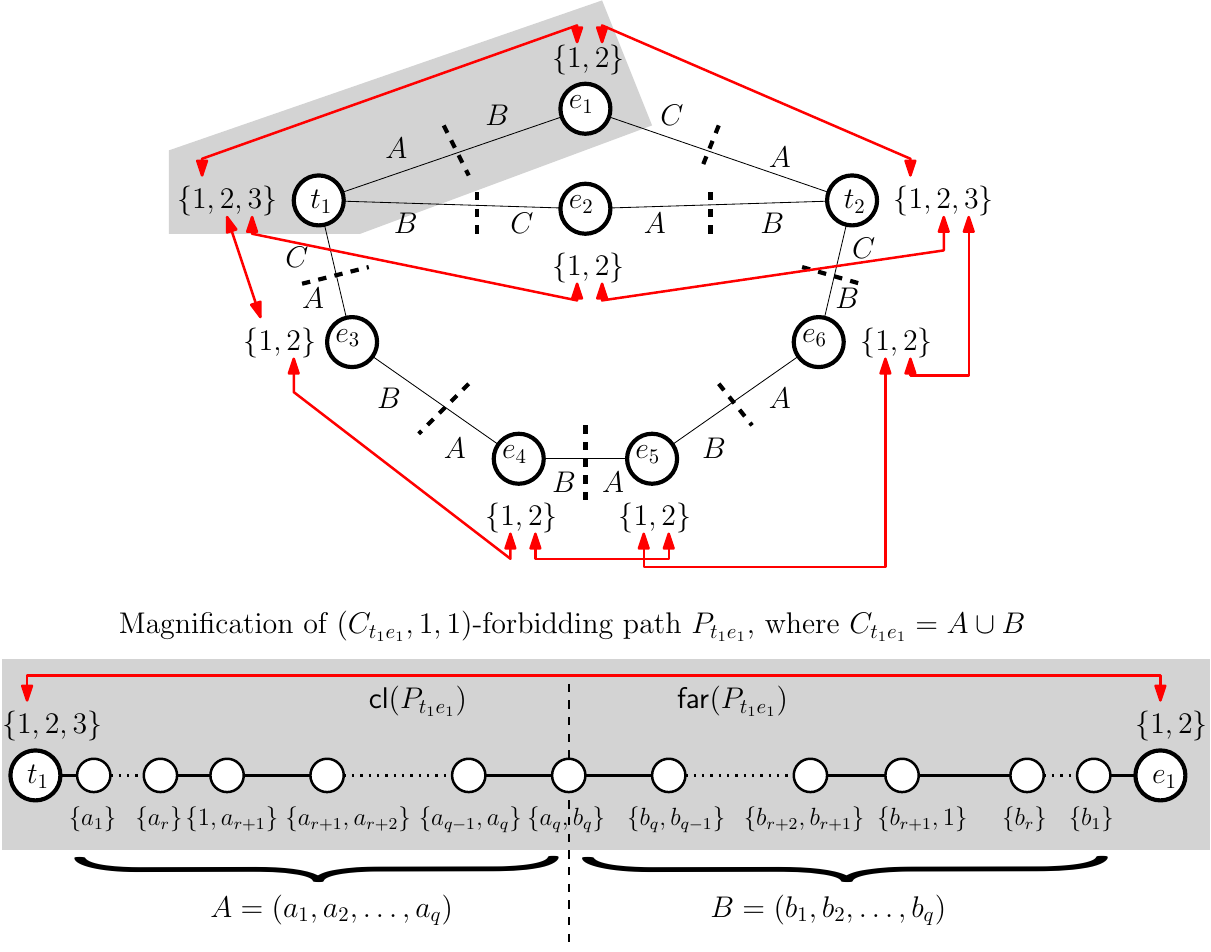}
    \caption{Construction of graph $G^\prime$ with a color list $L$ for each vertex of $G^\prime$ from the \ST's instance in \cref{fig:restricted-ST} and the color sets $\{A, B, C\}$. The
        construction of $(C_{t_1e_1}, 1, 1)$-forbidding path is described in details. Other forbidding paths are constructed similarly. The arrows at the end of red lines
        point to $(a,b)$ in a $(C_{uv}, a, b)$-forbidding path, where $u, v \in \{t_i, e_j\}$. }\label{fig:list-dk-col}
\end{figure}

Recall that given an instance $(G, T_A, T_B)$ of \ST, we need to construct an instance $(G^\prime, \alpha, \beta, L)$
of \textsc{List $(d, k)$-\CR}.
Up to the present, given $G$, one can verify that we have constructed $G^\prime$ and a color list $L$ for each vertex of $G^\prime$ in
polynomial time.

\begin{definition}\label{def:canonical-extension}
Let $P_{uv}$ be a $(C_{uv},a,b)$-forbidding path in $G^\prime$.
For every admissible pair $(a',b')\in L(u)\times L(v)\setminus\{(a,b)\}$, we fix a specific $[a',b']$-coloring of $P_{uv}$,
denoted by $\mathrm{Can}_{uv}^{a',b'}$.
More precisely, we choose $\mathrm{Can}_{uv}^{a',b'}$ as the $[a',b']$-coloring obtained by following the explicit constructive choices in the proof of \cref{lem:forbidpath}
on the core path and keeping all buffer vertices fixed as in \cref{lem:forbidpath-padding}.
\end{definition}

\begin{lemma}\label{lem:canonical-valid}
For every forbidding path $P_{uv}$ and every admissible pair $(a',b')$, the coloring $\mathrm{Can}_{uv}^{a',b'}$ is a list $(d,k)$-coloring of $P_{uv}$.
\end{lemma}
\begin{proof}
The existence of a list $(d,k)$ $[a',b']$-coloring of the core path follows from \cref{lem:forbidpath},
and \cref{lem:forbidpath-padding} extends it to the padded path by keeping the buffer vertices fixed.
\end{proof}

We now describe how to construct a \textsc{List $(d, k)$-\CR} $\alpha$ of $G^\prime$ based on $T_A$ where $k$ is ${3d+3}$ if $d$ is odd and
${3d+6}$ if $d$ is even.
For each $x \in V(G^\prime)$,
\begin{itemize}
    \item If $x = t_i$ ($1 \leq i \leq n_t$), we define $\alpha(x) = a$ if $t_{ia} \in T_A$ where $a \in \{1,2,3\}$.
          Similarly, if $x = e_j$ ($1 \leq j \leq n_e$), we define $\alpha(x) = a$ if $e_{ja} \in T_A$ where $a \in \{1, 2\}$.

    \item If $x \in V(P_{uv} \setminus \{u, v\})$ for some $(C_{uv}, a, b)$-forbidding path $P_{uv}$ of $G$, let $(a^\prime,b^\prime):=(\alpha(u),\alpha(v))$.
          {Since $(a^\prime,b^\prime)\neq(a,b)$ is an admissible pair, we define $\alpha(x):=\mathrm{Can}_{uv}^{a^\prime,b^\prime}(x)$ using the canonical $[a^\prime,b^\prime]$-coloring from \cref{def:canonical-extension}.}
\end{itemize}

We can also safely assume that all pairs $(\alpha(u), \alpha(v))$, where $u,v \in \{t_i, e_j\} \subseteq V(G^\prime)$ corresponding to
token triangles and token edges of $G$, are admissible pairs.
This follows as a direct consequence of $\alpha$ being constructed from $T_A$.
The construction of $\beta$ based on $T_B$ can be done similarly.
The following lemma confirms that $\alpha$ and $\beta$ are indeed list $(d, k)$-colorings of $G^\prime$.

\begin{lemma}\label{lem:list-dkcol-const}
    $\alpha$ is a list $(d, k)$-coloring of $G^\prime$.
    Consequently, so is $\beta$.
\end{lemma}
\begin{proof}
    {In the graph $G^\prime$, the vertices whose color lists contain a color from $\{1,2,3\}$ are
    \begin{itemize}
        \item the skeleton vertices $t_i$ and $e_j$ corresponding to token triangles and token edges of $G$, and
        \item for each forbidding path $P_{uv}$, the two core vertices $w_{uv}^{r+1}$ and $w_{uv}^{p-(r+1)}$ (with lists $\{a,x_{r+1}\}$ and $\{y_{r+1},b\}$, respectively).
    \end{itemize} 
    In particular, all buffer vertices $w_{uv}^i$ with $1\le i\le r$ or $p-r\le i\le p-1$ have singleton lists avoiding $\{1,2,3\}$.
    For every edge $uv$ whose $u$ and $v$ are skeleton vertices, $G'$ contains a forbidding path $P_{uv}$ of length $p\ge d+3$ between $u$ and $v$.
    Consequently, any two distinct skeleton vertices are at distance at least $p>d$ in $G'$ (since any path between them traverses at least one whole forbidding path).

    It remains to check that no two vertices at distance at most $d$ share a base color in the specific coloring $\alpha$.
    By construction, the only vertices that may be colored by a base color in $\alpha$ are the skeleton vertices and (depending on the endpoint colors) the two core vertices $w_{uv}^{r+1}$ and $w_{uv}^{p-(r+1)}$ on each forbidding path $P_{uv}$.
    All buffer vertices have singleton lists disjoint from $\{1,2,3\}$.

    Any two distinct skeleton vertices are at distance at least $p>d$.
    Next, consider a skeleton vertex $u$ and an incident forbidding path $P_{uv}$.
    In the canonical coloring $\mathrm{Can}_{uv}^{\alpha(u),\alpha(v)}$, the core vertex $w_{uv}^{r+1}$ is colored by the forbidden endpoint-color $a$ only when $\alpha(u)\neq a$, and it is colored by the auxiliary color $x_{r+1}$ when $\alpha(u)=a$.
    In particular, $w_{uv}^{r+1}$ never receives the same base color as $u$ in $\alpha$ (and symmetrically at the $v$-end).
    Finally, two core vertices lying on different paths incident to the same skeleton vertex are at distance at least $2(r+1)>d$, while the two core vertices of a single forbidding path are at distance
    $p-2(r+1)\in\{d+1,d+2\}>d$.
    Therefore, no two vertices within distance at most $d$ share a base color in $\alpha$.
    }

    {Next, consider any vertex $z\in V(G^\prime)$ whose list contains a color from $A\cup B\cup C$, and write $\alpha(z)=\chi$.
Then $z$ lies on some forbidding path $P_{uv}$ and $\chi\in L(z)\subseteq C_{uv}$.
By construction, $\alpha$ restricted to $P_{uv}$ is the canonical coloring $\mathrm{Can}_{uv}^{\alpha(u),\alpha(v)}$,
and hence it is a list $(d,k)$-coloring of $P_{uv}$ by \cref{lem:canonical-valid}.
In particular, no vertex of $P_{uv}$ at distance at most $d$ from $z$ is colored $\chi$.

Now let $z'\notin V(P_{uv})$ be any vertex with $\alpha(z')=\chi$.
Then necessarily $\chi\in L(z')$ and $z'$ lies on a forbidding path $P_{u'v'}\neq P_{uv}$.
\cref{lem:palette-separation} yields $\dist_{G^\prime}(z,z')>d$.
Therefore, no vertex at distance at most $d$ from $z$ shares color $\chi$ in $\alpha$.}

\end{proof}

Next, let us show how to efficiently construct such a function $f$.
\begin{lemma}\label{lem:fmapping-exists}
    Let $A,B,C$ be three disjoint sets where each set $A$, $B$ or $C$ is an ordered set of colors of size {$d$} if $d$ is odd or {$d+1$} if $d$ is even. Then we can in polynomial
    time construct
    $f\colon \part(P_{uv}) \rightarrow \{A,B,C\}$ that fulfill~\ref{item:forbidden-color:a}
    and~\ref{item:forbidden-color:b}.
\end{lemma}

\begin{proof}
We construct $f$ via a $3$-edge-coloring of an auxiliary bipartite graph.
Let $H$ be the graph obtained from the skeleton of $G^\prime$ by subdividing every forbidding path once:
the vertex set of $H$ consists of all skeleton vertices $\{t_i\mid 1\le i\le n_t\}\cup\{e_j\mid 1\le j\le n_e\}$
together with one subdivision vertex $m_{uv}$ for each forbidding path $P_{uv}$ in $G^\prime$.
For each $P_{uv}$ we add edges $um_{uv}$ and $vm_{uv}$ to $H$.

Then $H$ is bipartite (skeleton vertices on one side and subdivision vertices on the other),
and $\Delta(H)\le 3$ because every skeleton vertex has degree at most $3$ in $G^\prime$ and each $m_{uv}$ has degree $2$.
Since the edge chromatic number of any bipartite graph equals its maximum vertex degree (K\H{o}nig's line coloring theorem), every bipartite graph admits a proper edge-coloring with $\Delta(H)$ colors.
Hence $H$ has a proper edge-coloring $\varphi\colon E(H)\to\{A,B,C\}$.
Moreover, such an edge-coloring can be found in polynomial time, for example by repeatedly extracting maximum matchings.

We now define $f$ from $\varphi$.
For each forbidding path $P_{uv}$, recall that $\part(P_{uv})=\{\cl(P_{uv}),\far(P_{uv})\}$,
where $\cl(P_{uv})$ is the part incident to $u$ and $\far(P_{uv})$ is the part incident to $v$.
We set
\[
f(\cl(P_{uv})):=\varphi(um_{uv})\qquad\text{and}\qquad f(\far(P_{uv})):=\varphi(vm_{uv}).
\]

Since $\varphi$ is proper, the edges incident to any skeleton vertex receive pairwise distinct colors.
Therefore, the three parts incident to any degree-$3$ skeleton vertex are assigned pairwise distinct members of $\{A,B,C\}$,
which is Condition~\ref{item:forbidden-color:b}.
Furthermore, at each subdivision vertex $m_{uv}$ the two incident edges receive distinct colors,
so $f(\cl(P_{uv}))\neq f(\far(P_{uv}))$ for every $P_{uv}$, which is Condition~\ref{item:forbidden-color:a}.
This completes the construction of $f$.
\end{proof}

We are now ready to show the correctness of our reduction.

\begin{lemma}\label{lem:list-dkcol-correct}
    $(G, T_A, T_B)$ is a yes-instance if and only if $(G^\prime, \alpha, \beta, L)$ is a yes-instance.
\end{lemma}
\begin{proof}
    We claim that there is a $\mathsf{TS}$-sequence $\mathcal{S}$ between $T_A$ and $T_B$ in $G$ if and only if there is a sequence of
    valid recoloring steps $\mathcal{R}$ between $\alpha$ and $\beta$ in $G^\prime$.
    \begin{itemize}
        
\item[($\Rightarrow$)] Let $\mathcal{S}$ be a $\mathsf{TS}$-sequence in $G$ between $T_A$ and $T_B$.
    We describe how to construct the desired reconfiguration sequence $\mathcal{R}$ from $\mathcal{S}$.
    More precisely, for each move $x \to y$ in $\mathcal{S}$, we construct a corresponding sequence of recoloring steps in
    $\mathcal{R}$ as follows.
    From our construction of \ST, it follows that both $x$ and $y$ must be in the same token triangle or token edge.
    We consider the case $x = t_{ia}$ and $y = t_{ib}$ where $a, b \in \{1, 2, 3\}$, i.e., they are in the same token triangle
    $i \in \{1, \dots, n_t\}$.
    The other case can be handled similarly.
    In this case, corresponding to this move, we wish to recolor $t_i$ (which currently has color $a$) by $b$.

    {Let $v_1,\dots,v_\ell$ be the neighbors of $t_i$ in $G^\prime$ (so $\ell\le 3$), and for each $j$ let $P_j:=P_{t_i v_j}$.
    Let $c_j$ denote the current color of $v_j$.
    Since $t_{ia}\to t_{ib}$ is a valid $\mathsf{TS}$-move, no token occupies any vertex adjacent to $t_{ib}$ in $G$.
    By the definition of \ST and the reduction, this implies that for every $j$, the endpoint pair $(b,c_j)$ is admissible for the forbidding path $P_j$.}

    {For each $j\in\{1,\dots,\ell\}$, we apply Definition~\ref{def:ab-forbid}(2) to $P_j$ with $u=t_i$, $v=v_j$, $x=a$, $y=c_j$, and $x^\prime=b$.
    This yields a reconfiguration sequence $\mathcal{R}_j$ that recolors only vertices of $P_j$, keeps $v_j$ fixed to $c_j$,
    and recolors $t_i$ only in its last step (from $a$ to $b$).
    Let $\mathcal{R}_j^{-}$ be the prefix of $\mathcal{R}_j$ obtained by deleting this last step.
    By Corollary~\ref{cor:path-locality}, we may execute $\mathcal{R}_j^{-}$ inside $G^\prime$ while keeping all vertices outside $P_j$ fixed.
    We execute $\mathcal{R}_1^{-},\mathcal{R}_2^{-},\dots,\mathcal{R}_\ell^{-}$ one after another (keeping $t_i$ colored $a$ throughout),
    and finally recolor $t_i$ from $a$ to $b$.
    This last step is valid because for each $j$ the current restriction to $P_j$ is exactly one step before the last step of $\mathcal{R}_j$,
    and hence no vertex of $P_j$ within distance at most $d$ from $t_i$ is colored $b$.
    Moreover, any vertex of $G^\prime$ not lying on one of the paths $P_1,\dots,P_\ell$ is at distance greater than $d$ from $t_i$
    (since any path from $t_i$ to another skeleton vertex traverses a forbidding path of length $p\ge 2d+2$).
    The resulting concatenation is a valid recoloring subsequence of $\mathcal{R}$ that simulates the move $t_{ia}\to t_{ib}$.}

    Combining these sequences give us our desired sequence $\mathcal{R}$.

\item[($\Leftarrow$)] Suppose that $\mathcal{R}$ is a sequence of valid recoloring steps between $\alpha$ and $\beta$.
            We construct our desired $\mathsf{TS}$-sequence between $T_A$ and $T_B$ from $\mathcal{R}$ as follows.
            For each recoloring step in $\mathcal{R}$, we construct a corresponding $\mathsf{TS}$-move in $\mathcal{S}$, which may
            sometimes be a redundant step that reconfigures a token-set to itself.
            We maintain the invariant that after processing each recoloring step, the current token-set in $G$ is standard and is encoded by the current colors of the skeleton vertices of $G^\prime$:
            for each token triangle $i$ the token is on $t_{i\gamma(t_i)}$, and for each token edge $j$ the token is on $e_{j\gamma(e_j)}$, where $\gamma$ denotes the current coloring of $G^\prime$.
            Suppose that $v \in V(G^\prime)$ is currently recolored.
            \begin{itemize}
                \item If $v$ is in a forbidding path $P_{xy}$ and $v \notin \{x, y\}$, we add a redundant step to $\mathcal{S}$.
                \item If $v$ is either $t_i$ ($1 \leq i \leq n_t$) or $e_j$ ($1 \leq j \leq n_e$), suppose that $v$ is recolored from
                      color $a$ to color $b$, where $a,b\in L(v)$.
                      We first consider the case $v = t_i$ (so $a,b\in\{1,2,3\}$).
                      {By the invariant, the current token configuration contains $t_{ia}$ and no other vertex of triangle $i$, so $t_{ib}$ is unoccupied.
                      Let $z$ be any neighbor of $t_{ib}$ in $G$ distinct from $t_{ia}$.
                      By construction of $G$ and $G^\prime$, this adjacency $t_{ib}z$ is represented by a forbidding path $P_{t_i w}$ in $G^\prime$ whose endpoints are $t_i$ and the unique gadget-vertex $w\in\{t_{i'}:i'\neq i\}\cup\{e_{j'}\}$ containing $z$.
                      Moreover, there is a unique color $c\in\{1,2,3\}$ such that $z$ is the vertex indexed by $c$ in the gadget of $w$, and the unique forbidden endpoint pair of $P_{t_i w}$ is $(b,c)$.
                      Let $c_w:=\gamma(w)$ be the current color of $w$ after recoloring $t_i$ to $b$.
                      Since this recoloring step is valid, the resulting coloring restricts to a $[b,c_w]$-coloring of $P_{t_i w}$; hence $(b,c_w)$ is admissible by Condition~\ref{def:ab-forbid:item1}.
                      Therefore $(b,c_w)\neq(b,c)$ and thus $c_w\neq c$, which implies (by the invariant) that the token of gadget $w$ is not placed on $z$.
                      As $z$ was arbitrary, no token is adjacent to $t_{ib}$ after the slide, and thus we can slide the token on $t_{ia}$ to $t_{ib}$ and add this step to $\mathcal{S}$.}

                      {Now suppose that $v=e_j$ for some $j\in\{1,\dots,n_e\}$, so $a,b\in L(e_j)=\{1,2\}$.
                      By the invariant, the current token configuration contains $e_{ja}$ and no other vertex of the edge gadget $j$, so $e_{jb}$ is unoccupied.
                      Let $z$ be any neighbor of $e_{jb}$ in $G$ distinct from $e_{ja}$.
                      By construction of $G$ and $G^\prime$, the adjacency $e_{jb}z$ is represented by a forbidding path $P_{e_j w}$ in $G^\prime$ whose endpoints are $e_j$ and the unique gadget-vertex $w\in\{t_{i'}\}\cup\{e_{j'}:j'\neq j\}$ containing $z$.
                      Moreover, there is a unique color $c_z\in L(w)$ such that $z$ is the vertex indexed by $c_z$ in the gadget of $w$, and the unique forbidden endpoint pair of $P_{e_j w}$ is $(b,c_z)$.
                      Let $c_w:=\gamma(w)$ be the current color of $w$ after recoloring $e_j$ to $b$.
                      Since this recoloring step is valid, the resulting coloring restricts to a $[b,c_w]$-coloring of $P_{e_j w}$; hence $(b,c_w)$ is admissible and $c_w\neq c_z$.
                      Therefore the token of $w$ is not placed on $z$.
                      As $z$ was arbitrary, no token is adjacent to $e_{jb}$ after the slide, and thus we can slide the token on $e_{ja}$ to $e_{jb}$ and add this step to $\mathcal{S}$.}
            \end{itemize}
    \end{itemize}
\end{proof}

Finally, we show \cref{thm:listdk-pspacec} as follows.

\thmlistdkpspacec*

\begin{proof}
    The \PSPACE-completeness and the values of $d$ and $k$ follows from our construction and proofs above.
    From our construction, since the input graph $G$ of a \ST is planar, so is the constructed graph $G^\prime$.
    As any forbidding path has even length and $G^\prime$ no longer contains any ``token triangle'', it follows that any cycle in
    $G^\prime$ has even length, and therefore it is also bipartite.
    Additionally, we can also show that $G^\prime$ is $2$-degenerate.
    Let us prove by contradiction. Let $X$ be an induced subgraph in $G^\prime$ such that the minimum degree of any vertex in $X$ is at
    least $3$. However, by construction of $G^\prime$ we know that for any vertex $x$ of degree $3$, all its neighbors have degree $2$. If
    $x \in X$, then at least one of its neighbors also belong to $X$ by definition. Hence, $X$ has a vertex with degree at most $2$ contradicting our
    assumption.
\end{proof}

Combining our construction in \cref{thm:coltolistcol} and proofs of \cref{thm:rst-pspacec,thm:listdk-pspacec}, we finally have \cref{thm:main}.

\thmmain*

\begin{proof}
    The \PSPACE-completeness follows from our construction and proofs above.
    Using \cref{thm:listdk-pspacec,thm:rst-pspacec}, one can assume without loss of generality that the input graph $G$ of a \textsc{List $(d,k)$-\CR} instance is planar,
    bipartite, and $2$-degenerate.
    As our constructed frozen graphs (\cref{thm:coltolistcol}, $F_v$, $v \in V(G)$) are trees when $d \geq 3$, they are also planar, bipartite and $1$-degenerate.
    Additionally, when $d = 2$, the constructed frozen graphs are planar and $2$-degenerate, but not bipartite (\cref{lem:dkcol-graph-type}).
    Combining both, we obtain that our constructed graph $G^\prime$ is planar, bipartite, and $2$-degenerate when $d \geq 3$, and planar and $2$-degenerate when $d = 2$.

    Finally, by construction, in \cref{sec:list-dk-col-reconf} we produced an equivalent instance of \textsc{List $(d,k_0)$-\CR} where $k_0\in\{3d+3,\,3d+6\}$.
Applying the reduction of \cref{sec:dk-col-reconf} yields an equivalent instance of \textsc{$(d,k^\prime)$-\CR} with
\[
k^\prime \;=\; k_0+2+n(\lceil d/2\rceil-1),
\]
where $n:=|V(G)|$ is the number of vertices of the input list-instance graph $G$.
In particular, for fixed $d$ we have $k^\prime=\Theta(nd + k)$, and \cref{lem:dkcol-color-lb} shows that the linear dependence on $n$ is unavoidable for the present list-to-nonlist scheme.
\end{proof}

\subsection{The case $d=2$ and bipartiteness}\label{sec:d2-bipartite}

\cref{lem:dkcol-graph-type} shows that, although our reduction preserves planarity and $2$-degeneracy for all $d\ge 2$,
it preserves bipartiteness only when $d\ge 3$.
To briefly explain why such a limitation exists, we show that within the present list-to-nonlist interface for $d=2$,
any frozen list-enforcement gadget that remains planar and $2$-degenerate cannot also be bipartite.

\begin{lemma}\label{lem:d2-no-bipartite-frozen}
Fix $d=2$ and let $v$ be a vertex whose list $L(v)$ is a proper subset of $\{1,\dots,k\}$.
Consider the interface of \cref{sec:dk-col-reconf}, where $v$ is attached by a single edge to an anchor $c_v$ of a precolored frozen gadget $F_v$
whose coloring uses all colors.
If $F_v$ is required to be planar and $2$-degenerate, then $F_v$ cannot be bipartite.
\end{lemma}
\begin{proof}
We work in the case $d=2$ and use the frozen-graph construction from \cref{sec:dk-col-reconf}
(see also the explicit verification for $d=2$ in \cref{lem:frozengraph}).
For a vertex $v\in V(G)$, the gadget $F_v$ has vertex set
\[
V(F_v)=\{c_v,c_v^\star,w_{v,1},\dots,w_{v,k}\},
\]
and edges
\[
E(F_v)=\{c_vc_v^\star\}\ \cup\ \{c_v^\star w_{v,i}\mid 1\le i\le k\}\ \cup\ \{c_v w_{v,i}\mid i\notin L(v)\}.
\]
(The intended coloring $\alpha_v$ assigns pairwise distinct colors to all vertices of $F_v$; in particular,
$\alpha_v$ uses every color appearing in the gadget, which is what yields frozenness for $d=2$ as argued in \cref{lem:frozengraph}.)

We first note (for completeness) that $F_v$ is planar and $2$-degenerate.
Indeed, $F_v$ is a subgraph of $K_{2,k}$ on bipartition
$\{c_v,c_v^\star\}$ and $\{w_{v,1},\dots,w_{v,k}\}$, together with the extra edge $c_vc_v^\star$.
Since $K_{2,k}$ is planar for all $k$, adding the edge $c_vc_v^\star$ preserves planarity (it can be drawn inside a face incident with both
$c_v$ and $c_v^\star$ in a standard planar drawing of $K_{2,k}$).
For $2$-degeneracy, remove the vertices $w_{v,1},\dots,w_{v,k}$ first:
each $w_{v,i}$ has degree $1$ if $i\in L(v)$ and degree $2$ if $i\notin L(v)$.
After deleting all $w_{v,i}$, the remaining graph is the single edge $c_vc_v^\star$.
Hence every nonempty induced subgraph of $F_v$ has a vertex of degree at most $2$, so $F_v$ is $2$-degenerate.

We now show that $F_v$ cannot be bipartite when $L(v)\subsetneq\{1,\dots,k\}$.
Since $L(v)$ is a proper subset, there exists some index $i\in\{1,\dots,k\}\setminus L(v)$.
By the definition of $E(F_v)$, we then have all three edges
\[
c_vc_v^\star\in E(F_v),\qquad c_v^\star w_{v,i}\in E(F_v),\qquad c_v w_{v,i}\in E(F_v).
\]
Therefore the vertices $c_v,c_v^\star,w_{v,i}$ span a $3$-cycle
\[
c_v \;-\; c_v^\star \;-\; w_{v,i} \;-\; c_v.
\]
In particular, $F_v$ contains an odd cycle, and thus $F_v$ is not bipartite.

Consequently, under the present $d=2$ list-to-nonlist interface (where forbidding a color $i\notin L(v)$ is implemented by adding the edge
$c_vw_{v,i}$ while keeping $c_v^\star$ adjacent to all $w_{v,i}$ to ensure frozenness),
any instance with at least one vertex $v$ satisfying $L(v)\subsetneq\{1,\dots,k\}$ forces a non-bipartite frozen gadget $F_v$.
This is exactly why, unlike the case $d\ge 3$ (where $F_v$ is essentially a tree and hence bipartite), the current $d=2$ scheme does not
preserve bipartiteness.

\end{proof}

\section{\PSPACE-Completeness on Split Graphs and Chordal Graphs}\label{sec:split}

In this section, we focus on split graphs and chordal graphs.
First, for completeness, we revisit the proof by Bodlaender et al.~\cite{BodlaenderKTL04} showing that \textsc{$(2, k)$-Coloring} is \NP-complete.
Then, we prove that $(2,k)$-\CR is \PSPACE-complete.
The case where $d \geq 3$, in which the problem can be solved in polynomial time, will be addressed in \cref{sec:ddiameter}.
We also extend the reduction to show that $(d,k)$-\CR is \PSPACE-complete on chordal graphs for any even value $d \geq 2$.

\begin{lemma}[\cite{BodlaenderKTL04}]\label{thm:2kcol-split}
   \textsc{$(2, k)$-Coloring} on split graphs is \NP-complete.
\end{lemma}
\begin{proof}
   One can verify that \textsc{$(2, k)$-Coloring} is in \NP: \textsc{$k$-Coloring} is in \NP{} and \textsc{$(2, k)$-Coloring} on a
   graph $G$ can be converted to \textsc{$k$-Coloring} on its square graph $G^2$.
   To show that it is \NP-complete, we describe a reduction from the well-known \textsc{$\ell$-Coloring} problem on general graphs for
   $\ell \geq 3$, which asks whether a given graph $G$ has a proper $\ell$-coloring.
   Let $(G, \ell)$ be an instance of \textsc{$\ell$-Coloring} where $G = (V, E)$ is an arbitrary graph.
   We construct an instance $(G^\prime, k)$ of \textsc{$(2, k)$-Coloring} where $G^\prime$ is a split graph as follows.
   To construct $G^\prime$, we first add all vertices of $G$ to $G^\prime$.
   For each edge $e = xy \in E(G)$ where $x, y \in V(G)$, we add a new vertex $v_e$ in $V(G^\prime)$. Corresponding to each edge
   $e = xy \in E(G)$, we add the edges $xv_e$ and $yv_e$ to $E(G^\prime)$.
   Between all vertices $\bigcup_{e \in E(G)}\{v_e\}$ we form a clique in $G^\prime$.
   Finally, we set $k = m + \ell$ where $m = \lvert E(G)\rvert$.
   Our construction can be done in polynomial time.

   From the construction, $G^\prime$ is a split graph with $K = \bigcup_{e \in E(G)}\{v_e\}$ forming a clique and $S = V(G)$ forming
   an independent set.
   We now prove that $G$ has a proper $\ell$-coloring if and only if $G^\prime$ has a $(2, k)$-coloring where $k = m + \ell$.
   \begin{itemize}
       \item[($\Rightarrow$)] Suppose that $G$ has a proper $\ell$-coloring $\alpha$.
           We construct a $(2, k)$-coloring $\alpha^\prime$ of $G^\prime$ by setting $\alpha^\prime(u) = \alpha(u)$ for every
           $u \in V(G) = S$ and use $m$ new colors to color all $m$ vertices in $K$.
           From the construction, if $\dist_G(u, v) = 1$ for $u, v \in V(G) = S$ then $\dist_{G^\prime}(u, v) = 2$.
           Thus, $\alpha^\prime$ is a $(2, k)$-coloring of $G^\prime$.

       \item[($\Leftarrow$)] Suppose that $G^\prime$ has a $(2, k)$-coloring $\alpha^\prime$.
           We construct a coloring $\alpha$ of vertices of $G$ by setting $\alpha(u) = \alpha^\prime(u)$ for every $u \in S = V(G)$.
           Observe that any pair of vertices in $K$ have different colors.
           Therefore, $\alpha^\prime$ uses $k - \lvert K\rvert = k - m = \ell$ colors to color vertices in $S$.
           Additionally, if $uv \in E(G)$, we have $\dist_{G^\prime}(u, v) = 2$ and therefore
           $\alpha^\prime(u) \neq \alpha^\prime(v)$,
           which implies $\alpha(u) \neq \alpha(v)$.
           Thus, $\alpha$ is a proper $\ell$-coloring of $G$.
   \end{itemize}
\end{proof}

We are now ready to prove \cref{thm:2kcolreconf-split}.

\thmsplit*

\begin{proof}
   We reduce from $\ell$-\CR for $\ell \geq 4$.
   Let $(G, \alpha, \beta)$ be an instance of $\ell$-\CR where $\alpha$ and $\beta$ are two proper $\ell$-colorings of a graph $G$ having
   $n$ vertices and $m$ edges.

   First, we construct an instance $(\tilde{G}, \tilde{\alpha}, \tilde{\beta})$ of $\ell$-\CR where    $\tilde{\alpha}$ and $\tilde{\beta}$ are two proper $\ell$-colorings of a graph $\tilde{G}$ having
   $n + \ell$ vertices and $m + \ell(\ell-1)/2$ edges.
   $\tilde{G}$ is constructed by adding a new isolated clique $K_\ell$ on $\ell$ vertices to $G$.
   Let $\gamma$ be a fixed $\ell$-coloring of $K_\ell$ obtained by assigning for each vertex of $K_\ell$ a distinct color from $\{1, \dots, \ell\}$.
   We define $\tilde{\alpha}(v) = \alpha(v)$ if $v \in V(G)$ and $\tilde{\alpha}(v) = \gamma(v)$ if $v \in V(K_\ell)$.
   Similarly, we define $\tilde{\beta}(v) = \beta(v)$ if $v \in V(G)$ and $\tilde{\beta}(v) = \gamma(v)$ if $v \in V(K_\ell)$.
   Our construction can be done in polynomial time.
   As no vertex in $K_\ell$ can be recolored, it follows that any reconfiguration sequence between $\alpha$ and $\beta$ in $G$ can be regarded as a reconfiguration sequence between $\tilde{\alpha}$ and $\tilde{\beta}$ in $\tilde{G}$, and vice versa.
   Thus, $(G, \alpha, \beta)$ is a yes-instance if and only if $(\tilde{G}, \tilde{\alpha}, \tilde{\beta})$ is a yes-instance.

   Next, we construct an instance $(G^\prime, \alpha^\prime, \beta^\prime)$ of $(2, k)$-\CR where $k = \ell + m + \ell(\ell-1)/2 = m + \ell(\ell+1)/2$ and
   $\alpha^\prime$ and $\beta^\prime$ are $(2, k)$-colorings of $G^\prime$.
   Indeed, we use the same construction as in the first proof of \cref{thm:2kcolreconf-split} to construct $G^\prime$, except now that we use $(\tilde{G}, \tilde{\alpha}, \tilde{\beta})$ as the starting instance instead of $(G, \alpha, \beta)$.
   Again, the construction can be done in polynomial time.
   To see that it is correct, we note that unlike in the first proof, in this case, as no vertex of $K_\ell$ can be recolored in $G^\prime$, it follows that no vertex in $K$ can be recolored by a color in $C_S = \{1, \dots, \ell\}$. Thus, one can only recolor vertices of $V(G) = V(\tilde{G}) - V(K_\ell)$ in $G^\prime$.
   In other words, any reconfiguration sequence between $\alpha'$ and $\beta'$ in $G'$ can be regarded as a reconfiguration sequence between $\tilde{\alpha}$ and $\tilde{\beta}$ in $\tilde{G}$, and vice versa.
   Thus, $(G', \alpha', \beta')$ is a yes-instance if and only if $(\tilde{G}, \tilde{\alpha}, \tilde{\beta})$ is a yes-instance.
   This completes our proof.
\end{proof}

Indeed, we can further extend the proof of \cref{thm:2kcolreconf-split} to show the \PSPACE-completeness of $(d, k)$-\CR for even values of $d \geq 2$ on chordal graphs (\cref{thm:chordal}).

\thmchordal* 

\begin{proof}
   We reduce from $\ell$-\CR for $\ell \geq 4$.
   Let $(G, \alpha, \beta)$ be an instance of $\ell$-\CR where $\alpha$ and $\beta$ are two proper $\ell$-colorings of a graph $G$ having
   $n$ vertices and $m$ edges.

   Let $(\tilde{G}, \tilde{\alpha}, \tilde{\beta})$ be an instance of $\ell$-\CR constructed as described in the proof of \cref{thm:2kcolreconf-split}.
   (That is, the graph $\tilde{G}$ contains $G$ and an isolated clique $K_\ell$.)
   $\tilde{G}$ is a graph having $n + \ell$ vertices and $m + \ell(\ell-1)/2$ edges, and both $\tilde{\alpha}$ and $\tilde{\beta}$ are proper $\ell$-colorings of $\tilde{G}$.

   We construct an instance $(G', \alpha', \beta')$ of $(d, k)$-\CR on chordal graphs where $d \geq 2$ is even and $k = m + \ell(\ell+1)/2 + (n+\ell)(d-2)/2$ as follows.
   In particular, when $d = 2$, our construction is exactly the same as the one in the proof of \cref{thm:2kcolreconf-split}.

   First, for each vertex $v \in V(\tilde{G})$, we add two new vertices $v^1$ and $v^2$ of $v$ to $V(G')$ along with a path $P(v^1, v^2)$ of length exactly $(d-2)/2$ between $v^1$ and $v^2$.
   For each edge $e = xy \in E(\tilde{G})$ where $x, y \in V(\tilde{G})$, we add a new vertex $v_e$ in $V(G')$.
   Corresponding to each edge $e = xy \in E(\tilde{G})$, we add the edges $x^1v_e$ and $y^1v_e$ to $E(G')$.
   Between all vertices $\bigcup_{e \in E(\tilde{G})}\{v_e\}$ we form a clique in $G'$.
   Clearly, $G'$ can be seen as a graph obtained by attaching disjoint paths of length $(d-2)/2$ to a split graph and thus it is a chordal graph.
   
   Next, we set $k = m + \ell(\ell+1)/2 + (n+\ell)(d-2)/2$.
   Two $(d, k)$-colorings $\alpha'$ and $\beta'$ of $G'$ are defined as follows.
   Suppose that $C = C_S \cup C_K$ is the set of $k$ colors, where $C_S = \{1, \dots, \ell\}$, $C_K = \{\ell+1, \dots, \ell+m+\ell(\ell-1)/2+(n+\ell)(d-2)/2\}$, and the colors in $C_S$ are used in $\tilde{\alpha}$ and $\tilde{\beta}$ to color vertices of $\tilde{G}$.
   We set $\alpha'(v^2) = \tilde{\alpha}(v)$ and $\beta'(v^2) = \tilde{\beta}(v)$ for each $v \in V(\tilde{G})$.
   For each remaining uncolored vertex $w \in V(G')$, we color $w$ in both $\alpha'$ and $\beta'$ by the same color (i.e., $\alpha'(w) = \beta'(w)$) that is selected from some unused colors in $C_K$.
   One can verify that $\alpha'$ and $\beta'$ are indeed $(d, k)$-colorings of $G'$.
   Our construction can be done in polynomial time.

   One can also verify that only vertices $v^2$ which corresponds to $v \in V(G)$ can be recolored in $G'$ and the colors must come from $C_S$.
   Again, using a similar argument as in the proof of \cref{thm:2kcolreconf-split}, we can show that $(G, \alpha, \beta)$ is a yes-instance if and only if $(G', \alpha', \beta')$ is a yes-instance.
   This completes our proof.
\end{proof}

\section{Some Polynomial-Time Algorithms}\label{sec:algorithms}

In this section, we show that $(d, k)$-\CR can be solved in polynomial time on graphs of diameter at most $d$ (\cref{sec:ddiameter}) and on paths (\cref{sec:paths}).

\subsection{Graphs of Diameter At Most $d$}\label{sec:ddiameter}

This section is devoted to proving \cref{thm:diam-d}.

\thmdiamd*

   \begin{proof}
       Let $G$ be a $(d, k)$-colorable graph on $n$ vertices whose diameter is at most $d$.
       Since $G$ has diameter at most $d$, for any $(d, k)$-coloring $\alpha$, we have $\alpha(u) \neq \alpha(v)$ for every
       $u, v \in V(G)$. Thus, $n \leq k$.

       Now, if $n = k$, any instance $(G, \alpha, \beta)$ of $(d, k)$-\CR on $G$ is a no-instance as no vertex can be recolored.
       Otherwise, we will prove later that any instance $(G, \alpha, \beta)$ of $(d, k)$-\CR on $G$ is a yes-instance.
       The described procedure allows us to solve $(d, k)$-\CR on $G$ simply by comparing $n$ and $k$, which takes $O(\log n + \log k)$ time.

       It remains to show that when $n < k$, for any instance $(G, \alpha, \beta)$, one can construct a reconfiguration sequence between $\alpha$ and $\beta$.
       Observe that one can recolor any vertex with some extra color
       that does not appear in the
       current $(d, k)$-coloring (such a color always exists because $n < k$).
       This observation allows us to construct any target $(d, k)$-coloring $\beta$ from some source $(d, k)$-coloring $\alpha$
       using \cref{alg:ddiameter}.
       Since $n < k$, each step correctly produces a new $(d, k)$-coloring of $G$.
       It is also clear from the description that the construction runs in $O(n)$ time as we get closer to the coloring $\beta$ one color
       at a time.

       \begin{algorithm}
           \begin{algorithmic}[1]
               \Repeat
               \State Pick a vertex $v$ where $\beta(v)$ is an extra color that is not used in the current coloring.
               \If{such $v$ cannot be found}
               \LComment{$\beta$ is indeed obtained by permuting the colors used by the current coloring on the set $V(G)$}
               \State Arbitrarily pick any vertex $w$ and recolor it by any extra color.\label[line]{alg:ddiamter:pick-arbitrary}
               \LComment{Such an extra color always exists as $n < k$. In the next iteration, there
                   exists a vertex $v$ whose $\beta(v)$---the previous color of $w$ becomes an extra color}
               \EndIf
               \State Recolor $v$ by the color $\beta(v)$.\label[line]{alg:ddiamter:recolor}
               \Until{$\beta$ is obtained}
           \end{algorithmic}
           \caption{$d$-diameter algorithm: $n < k$\label{alg:ddiameter}}
       \end{algorithm}

   \end{proof}

   Recall that the diameter of a component of any split graph is at most $3$.
   The following corollary is straightforward.

   \begin{corollary}\label{cor:split}
       $(d, k)$-\CR can be solved in polynomial time on split graphs for any fixed integers $d \geq 3$ and $k \geq d+1$.
   \end{corollary}

\subsection{Paths}\label{sec:paths}

In this section, we prove \cref{thm:paths}.
We assume that a path $P$ on $n$ vertices is partitioned into $\lceil n/(d+1) \rceil$ disjoint
blocks of $d+1$
consecutive vertices (except possibly the last block, which can have less than $d+1$ vertices).
We denote by $v_{i, j}$ the $j$-th vertex in the $i$-th block of $P$ if it exists, for $1 \leq i \leq \lceil n/(d+1) \rceil$ and
$1 \leq j \leq d+1$.
In particular, $v_{1,1}$ is always an endpoint of $P$. Notice that $v_{i,1}$ and $v_{i,d+1}$ have distance $d$.

\begin{lemma}\label{lem:path:compact_path}
   Let $\alpha$ be any $(d, d+1)$-coloring of an $n$-vertex path $P$. Then, $\alpha(v_{i,j}) = \alpha(v_{i^\prime, j})$, where
   $1 \leq i < i^\prime \leq \lceil n/(d+1) \rceil$.
\end{lemma}
\begin{proof}
   It suffices to show that $\alpha(v_{i, j}) = \alpha(v_{i+1, j})$ for every $1 \leq i \leq \lceil n/(d+1) \rceil - 1$ and
   $1 \leq j \leq d+1$ such that $v_{i+1, j}$ exists.
   (If $i < \lceil n/(d+1) \rceil - 1$, $v_{i+1, j}$ always exists. If $i = \lceil n/(d+1) \rceil - 1$, $v_{i+1, j}$ may or
   may not exist.)
   Let $Q$ be the path between $v_{i, j}$ and $v_{i+1, j}$.
   Let $u$ be the neighbor of $v_{i, j}$ in $Q$.
   Similarly, let $v$ be that of $v_{i+1, j}$.
   By definition, the $uv$-path in $P$ has length exactly $d - 1$, and therefore its vertices are colored by exactly $d$ colors.
   Since at most $d + 1$ colors are available and $\alpha$ is a $(d, d+1)$-coloring, $\alpha(v_{i+1, j})$ cannot have any of the
   colors that were assigned to the $uv$-path.
   Hence, we have $\alpha(v_{i, j}) = \alpha(v_{i+1, j})$.
\end{proof}

From \cref{lem:path:compact_path}, it follows that if exactly $d + 1$ colors are available, one cannot recolor any vertex on a path
$P$. We have the following corollary.
\begin{corollary}\label{cor:path:no-sequence}
   Let $P$ be a path on $n$ vertices. Then any instance $(P, \alpha, \beta)$ with $\alpha\neq \beta$ of
   $(d, d+1)$-\CR is a no-instance.
\end{corollary}
\begin{proof}
   Let $\alpha$ be a $(d, d+1)$-coloring of $P$.
   It suffices to show that no vertex in $P$ can be recolored.
   Suppose to the contrary that there exists $v = v_{i, j}$ such that one can obtain a $(d, d+1)$-coloring $\alpha^\prime$ of $P$ from
   $\alpha$ by recoloring $v$, where $1 \leq i \leq \lceil n/(d+1) \rceil$ and $1 \leq j \leq d+1$.
   Since $P$ has diameter more than $d$, the first block of $P$ always has exactly $d+1$ vertices. None of them can be recolored, so
   $v \neq v_{1, j}$.
   On the other hand, by \cref{lem:path:compact_path} we have,
   $\alpha^\prime(v_{1, j}) = \alpha^\prime(v_{i, j}) = \alpha^\prime(v) \neq \alpha(v) = \alpha(v_{i, j}) = \alpha(v_{1, j})$.
   This implies that if we recolor $v_{i, j}$ we are also forced to recolor $v_{1, j}$.
   Thus, we have $v = v_{1, j}$, a contradiction.
\end{proof}

Next, using the two subsequent lemmas, we show that one extra color is enough to recolor the graph. First, \cref{lem:recol-to-d+1}
says that we can transform any $(d, k)$-coloring, where $k \geq d+2$ to some $(d, d+1)$-coloring.
Then, \cref{lem:path:sequence}
shows that if both source and target colorings are $(d,d+1)$-colorings and we have $k\geq d+2$ colors, we can recolor the
graph, thereby completing the picture.

\begin{lemma}\label{lem:recol-to-d+1}
   Let $P$ be a path on $n$ vertices.
   Let $\alpha$ be a $(d, k)$-coloring of $P$ for $k \geq d+2$.
   Then, there exists a $(d, d+1)$-coloring $\beta$ of $P$ such that $(P, \alpha, \beta)$ is a yes-instance of $(d, k)$-\CR.
   Moreover, one can construct in $O(n^2)$ time a reconfiguration sequence between $\alpha$ and $\beta$.
\end{lemma}
\begin{proof}
   \cref{algo:recol-to-d+1} describes how to construct a sequence $\mathcal{S}$ between $\alpha$ and some $(d, d+1)$-coloring $\beta$
   of
   $P$. Informally, the algorithm starts by using the colors of the second block to recolor vertices of the first block. Then, in each
   iteration of the algorithm (which corresponds to the outer \textbf{for}
   loop starting at Line 2   ), the algorithm uses the colors of the $i$th block
   to recolor vertices of the blocks $i-1$ to $1$ in that order. So, each iteration of the algorithm takes at most $O(n)$ time and,
   hence,
   \cref{algo:recol-to-d+1} runs in $O(n^2)$ time. Each vertex is recolored at most $O(\lceil n / (d+1) \rceil) $ times.

   Next, we show the correctness of our algorithm.
   We prove using induction on the length (i.e., the number of recoloring steps) $\ell \geq 1$ of $\mathcal{S}$ that $\mathcal{S}$ is
   indeed a reconfiguration sequence from $\alpha$ to $\beta$.
   Let $t \in \{1, \dots, d+1\}$ be the number of vertices in the last block of $P$, which may be less than $d+1$.
   Once $\mathcal{S}$ is a reconfiguration sequence, it follows directly from the algorithm that the resulting coloring $\beta$ is a
   $(d, d+1)$-coloring of $P$: In $\beta$, every block of $P$ will have its first $t$ vertices colored by the colors used in $\alpha$
   for
   all $t$ vertices in the last block and its last $d + 1 - t$ vertices colored by the colors used in $\alpha$ for the last
   $d + 1 - t$ vertices in the second-to-last block.
   If the last block has $t = d+1$ vertices then $d + 1 - t = 0$ and thus all colors used in $\beta$ are used by $\alpha$ for
   vertices in the last block of $P$.
   \begin{algorithm}
       \begin{algorithmic}[1]
           \Require $(P, \alpha)$ where $\alpha$ is a $(d, k)$-coloring of a path $P$ for some $k \geq d+2$
           \Ensure A reconfiguration sequence $\mathcal{S}$ between $\alpha$ and some $(d, d+1)$-coloring $\beta$ of $P$
           \State $\mathcal{S} \gets \langle \alpha \rangle$
           \For{$i$ from $2$ to $\lceil n/(d+1) \rceil$}\label{alg:recol-to-d+1:for}
           \For{$j$ from $1$ to $d+1$}
           \If{$v_{i, j}$ exists}
           \For{$p$ from $i - 1$ to $1$}
           \State $\alpha(v_{p, j}) \gets \alpha(v_{i, j})$ \Comment{This can also be seen as recoloring $v_{p, j}$ by the color
               $\alpha(v_{p+1, j})$}
           \State $\mathcal{S} \gets \mathcal{S} \cup \langle \alpha \rangle$
           \EndFor
           \EndIf
           \EndFor
           \EndFor
           \Return{$\mathcal{S}$}
       \end{algorithmic}
       \caption{Construction of a reconfiguration sequence that transforms any $(d, k)$-coloring where $k \geq d+2$ into a
           $(d, d+1)$-coloring in a path\label{algo:recol-to-d+1}}

   \end{algorithm}

   For the base case $\ell = 1$, the sequence $\mathcal{S} = \langle \alpha, \alpha_1 \rangle$ where $\alpha_1$ is obtained from
   $\alpha$
   by recoloring $v_{1, 1}$ with the color $\alpha(v_{2, 1})$ is indeed a reconfiguration sequence: Since $\alpha$ is a
   $(d, k)$-coloring ($k \geq d + 2$) of $P$, no vertex in the path between $v_{1,  1}$ and $v_{2, 1}$ is colored by
   $\alpha(v_{2, 1})$. Since, the distance between $v_{1, 1}$ and $v_{2, 1}$ is exactly $d+1$, they can share the same color
   $\alpha(v_{2, 1})$.

   Next, assume that the sequence $\mathcal{S}^\prime = \langle \alpha, \alpha_1, \dots, \alpha_\ell \rangle$ obtained from
   \cref{algo:recol-to-d+1} is indeed a reconfiguration sequence in $P$.
   We claim that the sequence $\mathcal{S} = \langle \alpha, \alpha_1, \dots, \alpha_\ell, \alpha_{\ell + 1} \rangle$ is also a
   reconfiguration sequence in $P$. Suppose to the contrary that it is not.
   From the construction, there exist two indices $i$ and $j$ such that $\alpha_{\ell+1}$ is obtained from $\alpha_\ell$ by recoloring
   $v_{i, j}$ with the color $\alpha_{\ell}(v_{i+1,  j})$.
   Since $\mathcal{S}^\prime$ is a reconfiguration sequence but $\mathcal{S}$ is not, the above recoloring step is not valid, i.e.,
   there is a vertex $w \in V(P)$ such
   that $\alpha_\ell(w) = \alpha_\ell(v_{i+1, j})$, $w \neq v_{i, j}$, and the distance between $w$ and $v_{i, j}$ is at most $d$.
   By the distance constraint and the assumption that $\alpha_{\ell}$ is a $(d, k)$-coloring, $w$ is on the path between $v_{1, 1}$
   and $v_{i+1, j}$.
   (Recall that the distance between $v_{i, j}$ and $v_{i+1, j}$ is exactly $d+1$.)
   Since $\alpha_\ell(w) = \alpha_\ell(v_{i+1, j})$, $w$ is not in the $(i+1)$-th block.
   Thus, $w$ is in either the $i$-th block or the $(i-1)$-th one.
   We complete our proof by showing that in each case, a contradiction happens.
   \begin{itemize}
       \item We consider the case that $w$ is in the $i$-th block, say $w = v_{i, j^\prime}$ for some
               $j^\prime \in \{1, \dots, d+1\} \setminus \{j\}$.
               If $j^\prime > j$ then $w = v_{i, j^\prime}$ is on the path between $v_{i, j}$ and $v_{i+1, j}$.
               Recall that the path between $v_{i, j}$ and $v_{i+1, j}$ has length exactly $d+1$.
               So if $w$ is on that path and note that $w \neq v_{i, j}$, the distance between $w$ and
               $v_{i+1, j}$ is at most $d$.
               Since $\alpha_\ell$ is a $(d, k)$-coloring, we must have $\alpha_\ell(w) \neq \alpha_\ell(v_{i+1, j})$, a contradiction.
               On the other hand, if $j^\prime < j$ then by the inductive hypothesis,
               $\alpha_\ell(w) = \alpha_\ell(v_{i, j^\prime}) = \alpha_\ell(v_{i+1, j^\prime}) = \alpha_\ell(v_{i+1, j})$ (follows from
               construction of \cref{algo:recol-to-d+1}) which contradicts the assumption that $\alpha_\ell$ is a
               $(d, k)$-coloring of $P$.

       \item We now consider the case that $w$ is in the $(i-1)$-th block, say $w = v_{i-1, j^\prime}$ for some
               $j^\prime \in \{1, \dots, d+1\}$.
               Since the distance between $w = v_{i-1, j^\prime}$ and $v_{i, j}$ is at most $d$, we have $j^\prime > j$.
               By the inductive hypothesis, we have $\alpha_\ell(w) = \alpha_\ell(v_{i-1, j^\prime}) = \alpha_\ell(v_{i, j^\prime})$
               (follows from construction of \cref{algo:recol-to-d+1}).
               On the other hand, since $j^\prime > j$, the vertex $v_{i, j^\prime}$ is on the path between $v_{i, j}$ and
               $v_{i+1, j}$, and
               thus $\alpha_\ell(w) = \alpha_\ell(v_{i, j^\prime}) \neq \alpha_\ell(v_{i+1, j})$, a contradiction.
   \end{itemize}
\end{proof}

\begin{lemma}\label{lem:path:sequence}
   Let $P$ be a path on $n$ vertices.
   Then any instance $(P, \alpha, \beta)$ of $(d, k)$-\CR where $k \geq d+2$ and both $\alpha$ and $\beta$ are $(d, d+1)$-colorings
   of $P$ is a yes-instance.
   In particular, there exists a linear-time algorithm that constructs a reconfiguration sequence between $\alpha$ and $\beta$.
\end{lemma}
\begin{proof}
   A slight modification of \cref{alg:ddiameter} allows us to construct a reconfiguration sequence between $\alpha$ and
   $\beta$ in $O(n)$ time.
   Recall that at least $d + 2$ colors can be used.
   We apply \cref{alg:ddiameter} to the first block of $d + 1$ consecutive vertices $v_{1, 1}, \dots, v_{1, j}, \dots, v_{1, d+1}$ in
   $P$ with only one small change:
   when a vertex $v_{1, j}$ is considered for recoloring 
   (in Lines 5 and 7   of \cref{alg:ddiameter}), instead of just recoloring $v_{1, j}$, we also recolor the $j$-th vertex (if it
   exists) in every other block, one
   vertex at a time. This can be done correctly, as when we are recoloring $v_{1, j}$ using an extra color, that extra color is not
   present
   in the current coloring. So we can recolor the $j$-th vertices of all other blocks as well with that extra color.
   From \cref{thm:diam-d} and \cref{lem:path:compact_path}, it follows that this modified algorithm always correctly produces a
   $(d, k)$-coloring of $P$ at each step, and
   the algorithm runs in $O(n)$ time.
\end{proof}

Combining \cref{cor:path:no-sequence} and \cref{lem:recol-to-d+1,lem:path:sequence}, we are now ready to prove \cref{thm:paths}.

\thmpaths*

\begin{proof}
   Let $(P, \alpha, \beta)$ be an instance of $(d, k)$-\CR on paths.
   If $k = d + 1$, return ``no'' (Corollary  \ref{cor:path:no-sequence}).
   Otherwise, ($k \geq d+2$), return ``yes''.
   The algorithm simply compares $k$ and $d + 1$ and therefore takes $O(\log k + \log d)$ time.

   It remains to describe how to construct a reconfiguration sequence between $\alpha$ and $\beta$ in a yes-instance.
   If $\alpha$ (resp.\ $\beta$) is not a $(d, d+1)$-coloring of $P$, use \cref{lem:recol-to-d+1} to reconfigure it into some
   $(d, d+1)$-coloring $\alpha^\prime$ (resp.\ $\beta^\prime$).
   Otherwise, just simply assign $\alpha^\prime \gets \alpha$ (resp.\ $\beta^\prime \gets \beta$).
   Use \cref{lem:path:sequence} to construct a reconfiguration sequence between $\alpha^\prime$ and $\beta^\prime$.
   Combining these sequences gives us a reconfiguration sequence between $\alpha$ and $\beta$.
\end{proof}

\section{Concluding Remarks}\label{sec:conclusion}

We proved $(d, k)$-\CR is \PSPACE-complete for $d \geq 2$ on graphs that are planar and $2$-degenerate (and also bipartite when $d \geq 3$). 
Additionally, on split graphs, there is an interesting dichotomy: the problem is $\mathsf{PSPACE}$-complete when $d = 2$ and $k$ is large but can be solved efficiently when $d \geq 3$ and $k \geq d+1$. 
For chordal graphs, we show that the problem is $\mathsf{PSPACE}$-complete for even values of $d \geq 2$. Finally, we design a quadratic-time algorithm to solve the problem on paths for any $d \geq 2$ and $k \geq d+1$.
Following the natural hierarchy of graph degeneracy, a logical next open direction is to examine $1$-degenerate graphs (forests). 
Notably, the complexity of $(d, k)$-\CR ($d \geq 2$) remains unresolved even for trees.
We have only partially addressed this question by developing a quadratic-time algorithm for paths (a subclass of trees).

\begin{credits}
\subsubsection{\ackname}
We thank the anonymous reviewers for their useful comments.
A part of this work was done when Niranka was at RIMS, Kyoto University and Duc was at the Vietnam Institute for Advanced Study in Mathematics (VIASM).
Duc A. Hoang's research was partially supported by the Vietnam National University, Hanoi under the project QG.25.07 ``A study on reconfiguration problems from algorithmic and graph-theoretic perspectives''.
\end{credits}

\bibliographystyle{splncs04}
\bibliography{refs.bib}

@inproceedings{Sharp07,
	title         = {Distance Coloring},
	author        = {Alexa Sharp},
	year          = {2007},
	booktitle     = {Proceedings of ESA 2007},
	publisher     = {Springer},
	series        = {LNCS},
	volume        = {4698},
	pages         = {510--521},
	doi           = {10.1007/978-3-540-75520-3_46}
}

@article{BonsmaC09,
	title         = {Finding Paths Between Graph Colourings: {PSPACE}-Completeness and Superpolynomial Distances},
	author        = {Paul S. Bonsma and Luis Cereceda},
	year          = {2009},
	journal       = {Theoretical Computer Science},
	volume        = {410},
	number        = {50},
	pages         = {5215--5226},
	doi           = {10.1016/j.tcs.2009.08.023}
}

@article{JohnsonKKPP16,
	title         = {Finding Shortest Paths Between Graph Colourings},
	author        = {Matthew Johnson and Dieter Kratsch and Stefan Kratsch and Viresh Patel and Dani{\"{e}}l Paulusma},
	year          = {2016},
	journal       = {Algorithmica},
	volume        = {75},
	number        = {2},
	pages         = {295--321},
	doi           = {10.1007/s00453-015-0009-7}
}

@article{CerecedaHJ11,
	title         = {Finding Paths Between $3$-Colorings},
	author        = {Luis Cereceda and Jan {van den Heuvel} and Matthew Johnson},
	year          = {2011},
	journal       = {Journal of Graph Theory},
	volume        = {67},
	number        = {1},
	pages         = {69--82},
	doi           = {10.1002/jgt.20514}
}

@article{Wrochna18,
	title         = {Reconfiguration in Bounded Bandwidth and Treedepth},
	author        = {Marcin Wrochna},
	year          = {2018},
	journal       = {Journal of Computer and System Sciences},
	volume        = {93},
	pages         = {1--10},
	doi           = {10.1016/j.jcss.2017.11.003}
}

@article{HatanakaIZ19,
	title         = {The Coloring Reconfiguration Problem on Specific Graph Classes},
	author        = {Tatsuhiko Hatanaka and Takehiro Ito and Xiao Zhou},
	year          = {2019},
	journal       = {IEICE Transactions on Information and Systems},
	volume        = {E102.D},
	number        = {3},
	pages         = {423--429},
	doi           = {10.1587/transinf.2018FCP0005}
}

@article{BonsmaP19,
	title         = {Using Contracted Solution Graphs for Solving Reconfiguration Problems},
	author        = {Paul S. Bonsma and Dani{\"{e}}l Paulusma},
	year          = {2019},
	journal       = {Acta Informatica},
	volume        = {56},
	pages         = {619--648},
	doi           = {10.1007/s00236-019-00336-8}
}

@inproceedings{vanderZanden15,
	title         = {Parameterized Complexity of Graph Constraint Logic},
	author        = {Tom C. {van der Zanden}},
	year          = {2015},
	booktitle     = {Proceedings of IPEC 2015},
	publisher     = {Schloss Dagstuhl -- Leibniz-Zentrum f{\"u}r Informatik},
	series        = {LIPIcs},
	volume        = {43},
	pages         = {282--293},
	doi           = {10.4230/LIPIcs.IPEC.2015.282}
}

@article{KramerK1969,
	title         = {Un probleme de coloration des sommets d'un graphe},
	author        = {Kramer, Florica and Kramer, Horst},
	year          = {1969},
	journal       = {CR Acad. Sci. Paris A},
	volume        = {268},
	number        = {7},
	pages         = {46--48}
}

@article{KramerK1969-2,
	title         = {Ein F{\"a}rbungsproblem der Knotenpunkte eines Graphen bez{\"u}glich der Distanz p},
	author        = {Kramer, Florica and Kramer, Horst},
	year          = {1969},
	journal       = {Rev. Roumaine Math. Pures Appl},
	volume        = {14},
	number        = {2},
	pages         = {1031--1038}
}

@article{KramerK2008,
	title         = {A Survey on the Distance-Colouring of Graphs},
	author        = {Kramer, Florica and Kramer, Horst},
	year          = {2008},
	journal       = {Discrete mathematics},
	publisher     = {Elsevier},
	volume        = {308},
	number        = {2-3},
	pages         = {422--426},
	doi           = {10.1016/j.disc.2006.11.059}
}

@article{Mccormick1983,
	title         = {Optimal Approximation of Sparse Hessians and Its Equivalence to a Graph Coloring Problem},
	author        = {McCormick, S. Thomas},
	year          = {1983},
	journal       = {Mathematical Programming},
	publisher     = {Springer},
	volume        = {26},
	number        = {2},
	pages         = {153--171},
	doi           = {10.1007/BF02592052}
}

@article{LinS1995,
	title         = {Algorithms for Square Roots of Graphs},
	author        = {Lin, Yaw-Ling and Skiena, Steven S},
	year          = {1995},
	journal       = {SIAM Journal on Discrete Mathematics},
	publisher     = {SIAM},
	volume        = {8},
	number        = {1},
	pages         = {99--118},
	doi           = {10.1137/S089548019120016X}
}

@inproceedings{LeN2010,
	title         = {Hardness Results and Efficient Algorithms for Graph Powers},
	author        = {Le, Van Bang and Nguyen, Ngoc Tuy},
	year          = {2010},
	booktitle     = {Proceedings of WG 2009},
	pages         = {238--249},
	doi           = {10.1007/978-3-642-11409-0_21}
}

@article{HatanakaIZ15,
	title         = {The List Coloring Reconfiguration Problem for Bounded Pathwidth Graphs},
	author        = {Tatsuhiko Hatanaka and Takehiro Ito and Xiao Zhou},
	year          = {2015},
	journal       = {IEICE Transactions on Fundamentals of Electronics, Communications and Computer Sciences},
	volume        = {E98.A},
	number        = {6},
	pages         = {1168--1178},
	doi           = {10.1587/transfun.E98.A.1168}
}

@incollection{Heuvel13,
	title         = {The Complexity of Change},
	author        = {Jan {van den Heuvel}},
	year          = {2013},
	booktitle     = {Surveys in Combinatorics},
	publisher     = {Cambridge University Press},
	series        = {London Mathematical Society Lecture Note Series},
	volume        = {409},
	pages         = {127--160},
	doi           = {10.1017/cbo9781139506748.005}
}

@article{Nishimura18,
	title         = {Introduction to Reconfiguration},
	author        = {Nishimura, Naomi},
	year          = {2018},
	journal       = {Algorithms},
	volume        = {11},
	number        = {4},
	pages         = {52},
	doi           = {10.3390/a11040052}
}

@incollection{MynhardtN19,
	title         = {Reconfiguration of Colourings and Dominating Sets in Graphs},
	author        = {Mynhardt, C.M. and Nasserasr, S.},
	year          = {2019},
	booktitle     = {50 years of Combinatorics, Graph Theory, and Computing},
	publisher     = {CRC Press},
	pages         = {171--191},
	doi           = {10.1201/9780429280092-10},
	editor        = {Fan Chung and Ron Graham and Frederick Hoffman and Ronald C. Mullin and Leslie Hogben and Douglas B. West},
	edition       = {1st}
}

@article{HearnD05,
	title         = {{PSPACE}-Completeness of Sliding-Block Puzzles and Other Problems through the Nondeterministic Constraint Logic Model of Computation},
	author        = {Robert A. Hearn and Erik D. Demaine},
	year          = {2005},
	journal       = {Theoretical Computer Science},
	volume        = {343},
	number        = {1-2},
	pages         = {72--96},
	doi           = {10.1016/j.tcs.2005.05.008}
}

@phdthesis{Cereceda2007,
	title         = {Mixing Graph Colourings},
	author        = {Cereceda, Luis},
	year          = {2007},
	url           = {http://etheses.lse.ac.uk/131/},
	school        = {London School of Economics and Political Science}
}

@book{Diestel2017,
	title         = {Graph Theory},
	author        = {Reinhard Diestel},
	year          = {2017},
	publisher     = {Springer},
	series        = {Graduate Texts in Mathematics},
	volume        = {173},
	doi           = {10.1007/978-3-662-53622-3},
	edition       = {5th}
}

@phdthesis{Mahmoud2024,
	title         = {Graph Coloring Reconfiguration},
	author        = {Reem Mahmoud},
	year          = {2024},
	url           = {https://scholarscompass.vcu.edu/etd/7564/},
	school        = {Virginia Commonwealth University}
}

@article{BodlaenderKTL04,
	title         = {Approximations for $\lambda$-colorings of graphs},
	author        = {Hans L. Bodlaender and Ton Kloks and Richard B. Tan and Jan van Leeuwen},
	year          = {2004},
	journal       = {The Computer Journal},
	volume        = {47},
	number        = {2},
	pages         = {193--204},
	publisher     = {Oxford University Press},
	doi           = {10.1093/comjnl/47.2.193}
}

\end{document}